\RequirePackage{fix-cm}
\documentclass[smallextended]{svjour3} 
\smartqed 
\usepackage{graphicx}
\usepackage{amsfonts}
\usepackage{amsmath}
\makeatletter
\let\vec\@undefined
\makeatother
\usepackage{epstopdf}
\usepackage{url}
\usepackage{amsmath}
\newtheorem{observation}{Observation}
\usepackage{mathabx}
\usepackage[normalem]{ulem}

\usepackage{MnSymbol}
\usepackage[usenames,dvipsnames]{color}

\newcommand{\cT}{{\mathcal T}}
\newcommand{\cL}{{\mathcal L}}

\newcommand{\cF}{{\mathcal F}}
\newcommand{\cG}{{\mathcal G}}
\newcommand{\cX}{{\mathcal X}}

\newcommand{\mrca}{{\rm lca}}

\def\ball{\scalebox{0.5}{\rule[0.14cm]{0.15cm}{0.06cm}}}
\def\cherryA{\mathrel{\raisebox{-2pt}{$\ball$\hspace*{-0.5mm}}}}
\def\cherryB{\hspace*{-1.3mm}\mathrel{\raisebox{-2pt}{$\ball$}}}
\def\wedgeB{\ensuremath{\scalebox{1.4}{$\wedge$}}}
\def\fc{$\ensuremath{\wedge}$}

\def\fcL{\cherryA\joinrel$\wedgeB$\joinrel\hspace{0.8mm}\cherryB}

\usepackage[linesnumbered,ruled,vlined,english]{algorithm2e}

\begin{document}

\title{A first step towards computing all hybridization networks for two rooted binary phylogenetic trees
}


\author{\mbox{
 Celine Scornavacca$^{*}$  \and Simone Linz$^{*}$ \and
 Benjamin Albrecht} 
}

\authorrunning{Celine Scornavacca \and Simone Linz \and Benjamin Albrecht} 
\institute{
  C. Scornavacca \at
   Center for Bioinformatics (ZBIT), T\"ubingen University, Sand 14, 72076 T\"ubingen, Germany.\\
   Tel.: +49 (0)70 71 29 70455\\
   Fax: +49 (0)70 71 29 5148\\
   \email{scornava@informatik.uni-tuebingen.de}
\and
S. Linz \at
   Center for Bioinformatics (ZBIT), T\"ubingen University, Sand 14, 72076 T\"ubingen, Germany.\\
   Tel.: +49 (0)70 71 29 70455\\
   Fax: +49 (0)70 71 29 5148\\
   \email{linz@informatik.uni-tuebingen.de}   
   \and
   B. Albrecht \at
   Center for Bioinformatics (ZBIT), T\"ubingen University, Sand 14, 72076 T\"ubingen, Germany.\\
   \email{benjamin.albrecht@student.uni-tuebingen.de}\\[2mm]
      $^{*}$Equally contributing authors.
}

\titlerunning{Towards computing all hybridization networks for two rooted phylogenies}

\date{Received: date / Accepted: date}

\maketitle

\begin{abstract}
Recently, considerable effort has been put into developing fast algorithms to reconstruct a rooted phylogenetic network that explains two rooted phylogenetic trees and has a minimum number of {hybridization} vertices. With the standard approach to tackle this problem being combinatorial, the reconstructed network is rarely unique. From a biological point of view, it is therefore of importance to not only compute one network, but {\it all} possible networks. In this paper, we make a first step towards approaching this goal by presenting the first algorithm---called {\sc allMAAFs}---that calculates all maximum-acyclic-agreement forests for two rooted binary phylogenetic trees on the same set of taxa. 

\keywords{Directed acyclic graphs\and Hybridization\and Maximum-acyclic-agreement forests\and Bounded search\and Phylogenetics }
\end{abstract}

\section{Introduction}
\label{sec:introall}
Over the last decade, 
significant progress in phylogenetic studies has been achieved by combining the expertise 
acquired in the fields of biology, computer science, and mathematics. As for the latter, combinatorics is becoming increasingly important in approaching many problems in the context of reticulate evolution (e.g., see~\cite{huson2011survey,Semple2007} for two excellent reviews) which is an umbrella term for processes such as horizontal gene transfer, hybridization, and recombination. To analyze reticulation in evolution, the graph-theoretic concept of an agreement forest for two rooted phylogenetic trees has  attracted much attention (e.g.~\cite{BaroniEtAl2005,quantifyingreticulation,whiddenWABI,ChenWang2010,ASH}). 
However, most approaches that make use of this concept aim at quantifying the amount of reticulation that is needed to simultaneously explain a set of rooted phylogenetic trees. Thus, one is primarily interested in the number of horizontal gene transfer, hybridization, or recombination events that occurred during the evolution of a set of present-day species. 
Consequently, these approaches do not explicitly construct a rooted phylogenetic network that
 explains a set of phylogenetic trees.
 Nevertheless,
this is desirable from a biological point of view because such a network intuitively indicates how species may have evolved by means of speciation and reticulation. While each vertex of a phylogenetic tree has exactly one direct ancestor, a vertex of a phylogenetic networks may have more than one such ancestor; thereby indicating that the genome of the underlying species is a combination of the genomes of distinct parental species. Generically, we refer to such a vertex as a reticulation vertex or, more specific in the context of hybridization, as a hybridization vertex. Since reticulation events are assumed to be significantly less frequent than speciation events, current research aims at constructing a rooted phylogenetic network that explains a set of rooted phylogenetic trees and whose number of reticulation vertices is minimized.

For the purpose of the introduction, think of a so-called {\it maximum-acyclic-agreement forest} $\cF$ for two rooted binary phylogenetic trees $S$ and $T$ as a small collection of vertex-disjoint rooted subtrees that are common to $S$ and $T$ (for details, see Section~\ref{sec:prelim}). It is well-known that the size of $
\cF$ minus 1 equates to the minimum number of {hybridization} events that are needed to explain $S$ and $T$~\cite{BaroniEtAl2005}. {Furthermore, there exists an algorithm---called {\sc HybridPhylogeny}~\cite{BSS06}---that 
glues 
together 
the elements of $\cF$ by introducing new edges
 such that the resulting graph is a rooted phylogenetic network that explains $S$ and $T$ and has $|\cF|-1$ hybridization vertices}. However, until now, {\sc HybridPhylogeny}, has not found its way into many practical applications that are concerned with reconstructing the evolutionary history for a set of species whose past is likely to include {hybridization}. This might be due to the fact that the reconstructed phylogenetic network is rarely unique because the 
 gluing 
 step can often be done in a number of different ways. Furthermore, given two rooted binary phylogenetic trees $S$ and $T$, a maximum-acyclic-agreement forest for $S$ and $T$ is rarely unique. Given these hurdles, an appealing open problem is the reconstruction of {\it all} rooted phylogenetic networks that explain a pair of rooted phylogenetic trees and whose number of hybridization vertices is minimized. Once having calculated the entire solution space of these networks, one can then for example apply statistical methods or additional biological knowledge to decide which of the phylogenetic network in this space is most likely to be the correct one. 

In this paper, we focus on a first step to reach this goal. In particular, we give the first non-naive algorithm---called {\sc allMAAFs}---that is based on a bounded-search type idea and calculates all maximum-acyclic-agreement forests for two rooted binary phylogenetic trees $S$ and $T$ on the same set of taxa. With the underlying optimization problem being NP-hard~\cite{bordewich} and fixed-parameter tractable~\cite{sempbordfpt2007}, the running time of {\sc allMAAFs} is exponential. {More precisely}, we will see in Section~\ref{sec:RT} that the running time of of this algorithm is $O(3^{14k}+ p(n))$, where 
$n$ is the number of leaves in $S$ and $T$, $p(n)$ is some polynomial function that only depends on $n$, and $k$ is the minimum number of {hybridization} events needed to explain $S$ and $T$.  

The paper is organized as follows. The next section contains preliminaries and some well-known results from the phylogenetics literature. Section~\ref{sec:alg} describes the algorithm {\sc allMAAFs} that calculates all maximum-acyclic-agreement forests for two rooted binary phylogenetic trees. Its pseudocode is also given in this section. Subsequently, in Section~\ref{sec:proof}, we establish the correctness of {\sc allMAAFs} and give its running time in Section~\ref{sec:RT}.  We finish the paper with some concluding remarks in Section~\ref{sec:conclu}. 

\section{Preliminaries} \label{sec:prelim}
In this section, we give some preliminary definitions that are used throughout this paper. Notation and terminology on phylogenetic trees and networks follow~\cite{SempleSteel2003} and~\cite{HusonRuppScornavacca10}, respectively.\\

{\bf Phylogenetic trees.}
A {\it rooted  phylogenetic $\cX$-tree} $T$ is a connected graph with no (undirected) cycle, no vertices of degree 2, except for the root which has degree at least 2, and such that each element of $\cX$ labels a leaf of $T$.  The set $\cX$ represents a collection of present-day taxa and internal vertices  represent putative speciation events. 
A rooted phylogenetic $\cX$-tree $T$ is said to be  {\it binary} if its root vertex has degree two while
all other interior vertices have degree three. We denote the edge set of $T$ by $E(T)$.  The taxa set $\cX$ of $T$ is called the {\it label set} of $T$  and is frequently denoted by $\cL(T)$. 
Furthermore, let $v$ be a vertex of $T$. We denote by $\cL(v)$ the label set of the rooted phylogenetic tree with root $v$ that has been obtained from $T$ by deleting the edge ending in $v$.
Lastly, let $\cF$ be a set of rooted phylogenetic trees. Similarly to  $\cL(T)$, we use $\cL(\cF)$ to denote the union of leaf labels over all elements in $\cF$.

We next introduce several types of subtrees that will play an important role in this paper.
Let $T$ be a rooted phylogenetic $\cX$-tree, and let $\cX'\subset \cX$ be a subset of 
$\cX$.  We use $T(\cX')$ to denote the minimal connected subgraph of $T$ that
contains all leaves that are labeled by elements of $\cX'$.  Furthermore, the 
{\em restriction of $T$ to $\cX'$}, denoted by $T|_{\cX'}$, is defined as 
the rooted phylogenetic tree that has been obtained from $T(\cX')$ by suppressing all non-root degree-2 vertices. 
Lastly, we say that a subtree of $T$ is {\it pendant} if it can be detached from $T$ by deleting a single edge.\\

Now, let $T$ be a rooted binary phylogenetic $\cX$-tree, and let $\cX'$ be a subset of $\cX$. 
Then, the {\em lowest  common ancestor} of $\cX'$ in $T$ is the vertex $v$ in $T$ with $\cX'\subseteq \cL(v)$ such that there exists no vertex  $v'$ in $T$ with $\cX'\subseteq \cL(v')$ and $\cL(v')\subset\cL(v)$. We denote $v$ by $\mrca_T(\cX')$.\\

{\bf Hybridization networks.}
Let $\cX$ be a finite set of taxa. A {\it rooted  phylogenetic network} on $\cX$ is a rooted acyclic digraph with no vertex of both indegree and outdegree one and whose leaves are bijectively labeled by elements of $\cX$. 
Since this paper is concerned with hybridization as a representative of reticulation, we will often refer to a phylogenetic network as a {\it hybridization network}.
Each  internal vertex of a hybridization network with indegree  1 represents a  putative speciation event while each  vertex with indegree of at least $2$ represents a hybridization event and, therefore, a species whose genome is a chimaera of its parents' genomes.
Generically, we call a vertex of the latter type a {\it hybridization vertex} and each edge that enters a hybridization vertex a {\it hybridization edge}.
 
To quantify the number of hybridization events, the {\it hybridization number} of $N$, denoted by $h(N)$,  is defined as
\[h(N) = \sum_{\substack{v\in V(N): \delta^-(v)>0}}(\delta^-(v)-1) = |E| - |V| + 1,\]
where $V(N)$ denotes the vertex set of $N$ and $\delta^-(v)$ the indegree of $v$.
Note that, if $N$ is a rooted  phylogenetic
tree, then $h(N) = 0$, and if $\delta^-(v)$ is at most 2 for each vertex $v\in V(N)$, then $h(N)$ is equal to the total number of hybridization vertices of $N$.

Now, let $N$ be a phylogenetic network on $\cX$, and let $T$ be a rooted binary phylogenetic
$\cX'$-tree with $\cX'\subseteq \cX$. We say that $T$ is {\it displayed} by $N$ if $T$ can be obtained from $N$ by
deleting a subset of its edges and any resulting degree-0 vertices, and then contracting
edges. Intuitively, if $N$ displays $T$, then
all of the ancestral relationships visualized by $T$ are visualized by $N$.
In the remainder of this paper, we will consider the case where $\cT$ is composed of two rooted binary phylogenetic trees.

Extending the
definition of the hybridization number to two rooted binary phylogenetic $\cX$-trees $S$ and
$T$, we set
$$h(S,T ) = \min\{h(N) : N \mbox{ is a hybridization network that displays }S \mbox{ and }T\}.$$
Calculating $h(S,T)$ for two rooted binary phylogenetic $\cX$-trees has been shown to be NP-hard~\cite{bordewich}.\\

\begin{figure}
\begin{tabular}{ccc}
\includegraphics[width = 5cm]{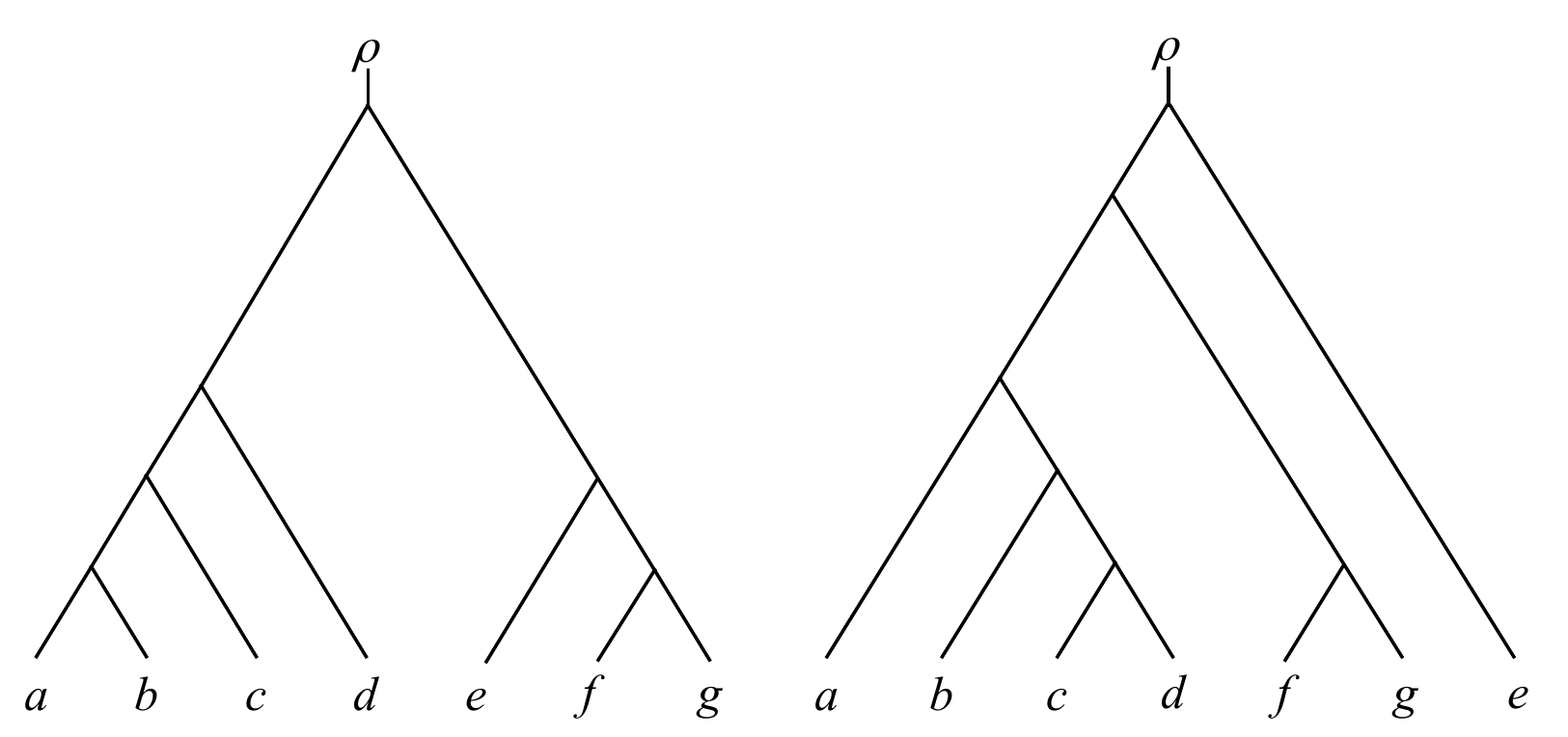}
&
\includegraphics[width = 3.0cm]{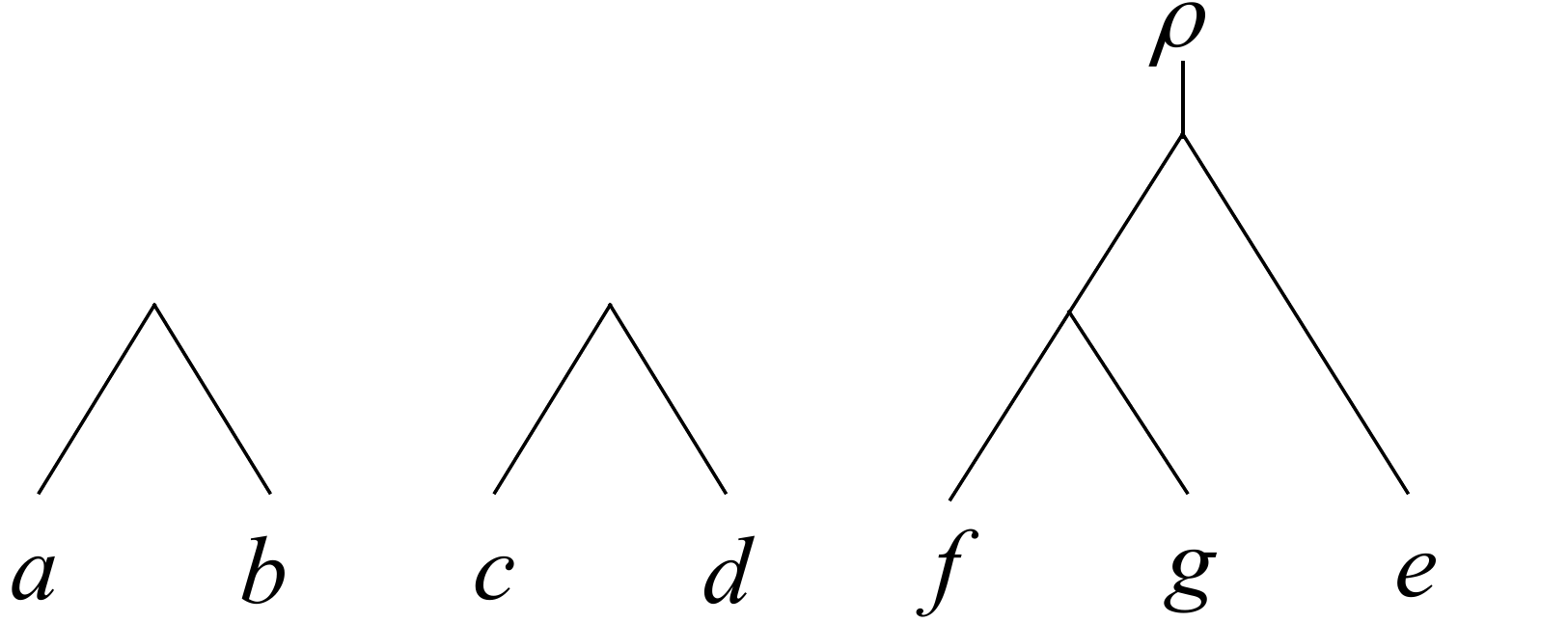}
&
\includegraphics[width = 3.0cm]{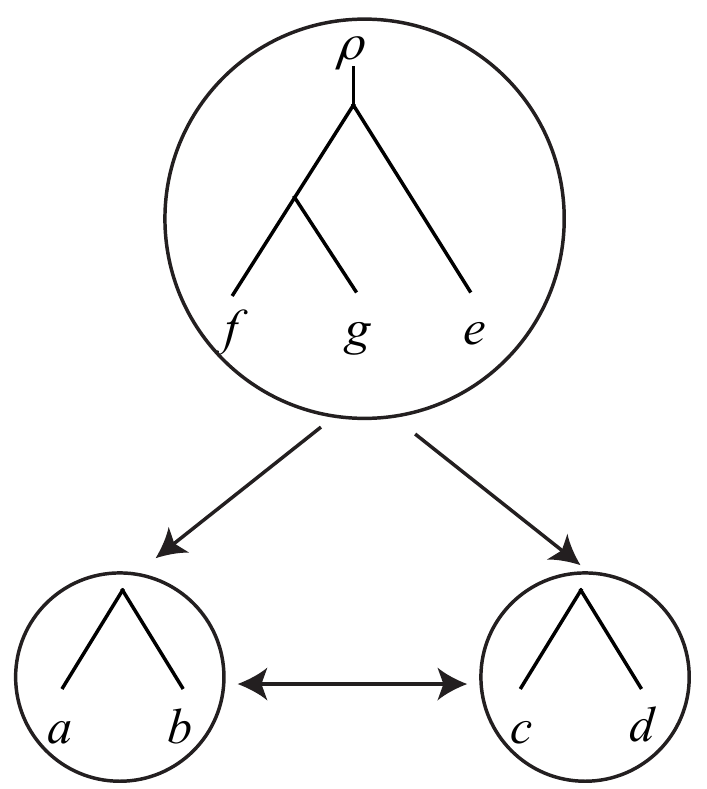}
\\
($i$)   & ($ii$)   & ($iii$)   \\
\end{tabular}
\caption{($i$) Two phylogenetic trees $S$ and $T$ on $\cX=\{a,b,c,d,e,f,g\}$. ($ii$) An agreement forest $\cF$ for $S$ and $T$. ($iii$) The graph $AG(S,T,\cF)$. Since this graph contains a {directed} cycle, $\cF$ is \emph{not} an acyclic-agreement forest for $S$ and $T$.}
\label{fig:forest}
\end{figure}

{\bf Forests.} Let $T$ be a rooted binary phylogenetic $\cX$-tree  whose edge set is $E(T)$. {For the purpose of the upcoming definitions and, indeed, much of the paper, we regard the root of $T$ as a vertex labeled $\rho$ at the end of a pendant edge adjoined to the original root of  $T$. For an example of two such trees, see Figure~\ref{fig:forest}(i).}
Furthermore, we view $\rho$ as an element of the label set of $T$; thus $\cL(T)= \cX\cup\{\rho\}$. Any collection of rooted binary phylogenetic trees whose union of label sets is {$\cL(T)$ is a {\it forest on $\cL(T)$}}.  Furthermore, we  say that a set  $\cF=\{F_0,F_1,\ldots,F_{k}\}$ of rooted binary phylogenetic trees, with $|\cF|$ referred to as the {\it size} of $\cF$, is a {\it forest for} $T$ if $\cF$ can be obtained from $T$ by deleting a $k$-sized subset $E$ of $E(T)$ and, subsequently, suppressing vertices with both indegree and outdegree 1. 
To ease reading, we write $\cF = T-E$ if $\cF$ can be obtained in this way. Obviously, in the same way, we obtain a new forest $\cF'=\{F_0',F_1',\ldots,F'_{k'}\}$ for $T$ from $\cF$ by deleting a $k'$-sized subset $E'$ of $ \bigcup_{F_i \in \cF} E(F_i)$ and, again, suppressing vertices with both indegree and outdegree 1. 
Similarly to the above, we write $\cF' = \cF -E'$ if $\cF'$ can be obtained in this way. Now, let $\cF$ be a forest for a rooted binary phylogenetic $\cX$-tree. We use $\overline{\cF}$ to denote the forest obtained from $\cF$ by deleting all of its
isolated vertices {and, additionally, the element that contains the vertex labeled $\rho$ if it contains at most one edge}.
Lastly, for two leaf vertices $a$ and $c$ with labels $\cL(a)$ and $\cL(c)$ 
respectively, we write $a\sim_\cF c$ if there exists an element in $\cF$ that 
contains a leaf labeled with $\cL(a)$ and a distinct leaf labeled with $\cL(c)$, 
otherwise, we write $a\nsim_\cF c$.\\

\begin{sloppypar}
Let $S$ and $T$ be two rooted binary phylogenetic $\cX$-trees. A set $\cF=\{F_\rho,F_1, F_2, \ldots,F_k\}$ of rooted phylogenetic trees is an {\it agreement forest} for $S$ and $T$ if $\cF$ is a forest for $S$ and $T$, and $\rho\in\cL(F_\rho)$. Note that the beforehand given definition is equivalent to the definition of an agreement forest that is usually used in the literature and that we give next. An {\it agreement forest} $\cF=\{F_\rho,F_1,F_2, \ldots,F_k\}$  for $S$ and $T$ is a collection of trees such that the following properties are satisfied:
\end{sloppypar}
\begin{enumerate}
\item[(i)] The label sets $\cL(F_\rho),\cL(F_1),\cL(F_2), \ldots, \cL(F_k)$ partition $\cX\cup\{\rho\}$ and, in particular, $\rho\in\cL(F_\rho)$.
\item [(ii)] For each $i\in\{\rho,1,2,\ldots,k\}$, we have $F_i\cong S|_{\cL(F_i)}\cong T|_{\cL(F_i)}$.
\item[(iii)] The phylogenetic trees in $\{S(\cL(F_i))\mid i\in\{\rho,1,2,\ldots,k\}\}$
and $\{T(\cL(F_i))\mid i=\{\rho,1,2,\ldots,k\}\}$ are vertex-disjoint subtrees of $S$ and $T$, respectively.
\end{enumerate}
Both definitions of an agreement forests for two rooted binary phylogenetic trees are used interchangeably throughout this paper.

An agreement forest with the minimum cardinality among all agreement forests for $S$ and $T$ is called a \emph{maximum-agreement forest} for $S$ and $T$. {An example of an agreement forest for the two trees $S$ and $T$ presented in Figure~\ref{fig:forest}(i), is shown in (ii) of the same figure. It is easy to check that this  forest is in fact a maximum-agreement forest for $S$ and $T$.}\

A characterization of the hybridization number $h(S,T)$ for two rooted binary phylogenetic trees $S$ and $T$ in terms of agreement forests requires an additional condition. Roughly, this condition avoids that species can inherit genetic material from their own offsprings. Let $\cF=\{F_\rho,F_1,F_2,\ldots,F_k\}$ be an agreement forest for two rooted binary phylogenetic $\cX$-trees $S$ and $T$. Furthermore, let
$AG(S,T,\cF)$ be the directed graph whose vertex set is $\cF$ and for which $(F_i,F_j)$ is an arc precisely if $i\ne j$,
and either
\begin{enumerate}
\item [(1)] the root of $S(\cL(F_i))$ is an ancestor of the root of  $S(\cL(F_j))$  in $S$, or
\item [(2)] the root of $T(\cL(F_i))$ is an ancestor of the root of  $T(\cL(F_j))$  in $T$.
\end{enumerate}
We call  $\cF$ an \emph{acyclic-agreement forest} for $S$ and $T$ if $AG(S,T,\cF)$ does not contain any directed cycle. {To illustrate, Figure~\ref{fig:forest}(iii) shows the graph $AG(S,T,\cF)$ for $S$, $T$, and $\cF$ of the same figure. Note that $\cF$ is not an acyclic-agreement forest for $S$ and $T$.}
Similarly to the definition of a maximum-agreement forest, an acyclic-agreement forest for $S$ and $T$ whose number of components is minimized over all such forests is called a \emph{maximum-acyclic-agreement forest} for $S$ and $T$. The importance of the concept of acyclic-agreement forests lies in the following theorem that has been established in~\cite[Theorem 2]{BaroniEtAl2005} and gives an attractive characterization of the hybridization number for two rooted binary phylogenetic trees.

\begin{theorem}
\label{t:hybrid}
Let $\cF=\{F_\rho,F_1,F_2,\ldots,F_k\}$ be a maximum-acyclic-agreement forest for two rooted binary phylogenetic $\cX$-trees $S$ and $T$. Then $$h(S,T)=k.$$
\end{theorem}

\noindent In the proof of Theorem~\ref{t:hybrid}, the authors implicitly show that, by deleting all hybridization edges of a  hybridization network $N$ that displays two rooted binary phylogenetic $\cX$-trees $S$ and $T$ and has a minimum number of hybridization vertices and, subsequently, suppressing all non-root degree-2 vertices, one obtains a maximum-acyclic-agreement forest $\cF$ for  $S$ and $T$. Note that $\cF$ is well-defined for when $N$ is given. We say that  $N$ {\it yields} $\cF$.  On the other hand, given a maximum-acyclic-agreement forest $\cF$ for two rooted binary phylogenetic $\cX$-trees $S$ and $T$, and using the algorithm {\sc HybridPhylogeny} (for details, see~\cite{BSS06}) to construct a hybridization network $N$ from $\cF$ that displays $S$ and $T$ and yields $\cF$, $N$ is rarely unique. Nevertheless, if one aims at reconstructing all hybridization networks that display  $S$ and $T$ and whose hybridization number is minimized, one can first calculate all maximum-acyclic-agreement forests for  $S$ and $T$ and then construct all possible minimum hybridization networks for each such forest.  As  mentioned in the introduction, this paper focuses on the first step of this approach, i.e. finding all maximum-acyclic-agreement forests for $S$ and $T$.\\

%

Now, let $\cF$ be a set of rooted binary phylogenetic trees, and let $a$ and $c$ be two distinct leaves of $\cF$. We say that $a$ and $c$ form a {\it cherry} in $\cF$ if they are adjacent to a common vertex, in which case we denote this cherry by $\{a,c\}$. 
Note that $a$ and $c$ refer to leaf vertices and not leaf labels. Let $S$ and $T$ be two rooted binary phylogenetic $\cX$-tree, and let $\cF$ be a forest for $T$. Furthermore, let $\{a,c\}$ be a cherry of $S|_{\cL(\overline{\cF})}$. 
We say that $\{a,c\}$ is a {\it contradicting cherry} of $S$ and $\cF$ if there is no cherry $\{a',c'\}$ in $\overline{\cF}$ such that one of $a'$ or $c'$, say $a'$, is labeled $\cL(a)$ while $c'$ is labeled $\cL(c)$. Otherwise, we call $\{a,c\}$ a \emph{common cherry} of $S$ and $\cF$.\\

\begin{figure}
\begin{center}
\begin{tabular}{c}
\includegraphics[width = 8cm]{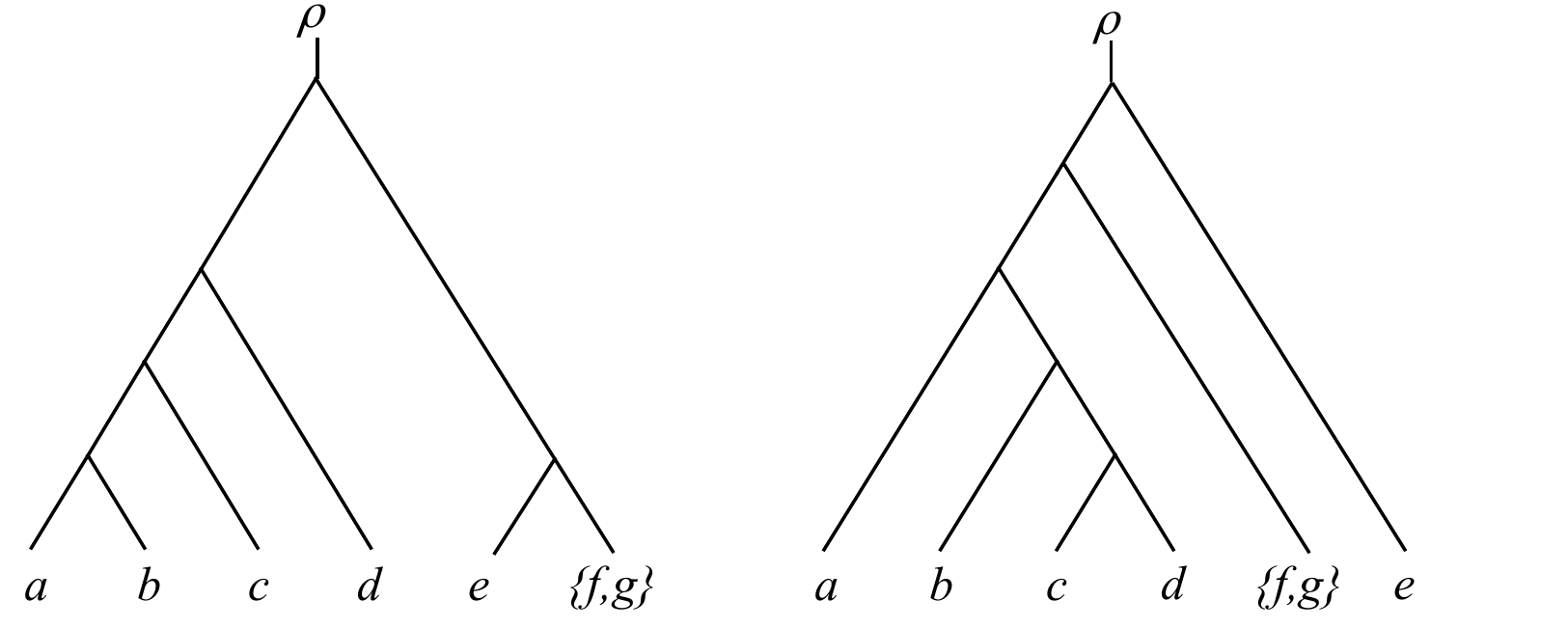}
\\
\end{tabular}
\caption{The two phylogenetic trees obtained by calling \textsc{cherryReduction}($S$, $T$, $\emptyset$, $\{f,g\}$), where $S$ and $T$ are the two phylogenetic trees shown Figure \ref{fig:forest}(i).}
\label{fig:cherry}
\end{center}
\end{figure}

{\bf Cherry reduction.} Let $\cF$ be a forest for a rooted binary phylogenetic tree, and let $\{a,c\}$ be a  cherry of $\cF$. The operation of deleting the two leaf vertices $a$ and $c$ and their respective labels and labeling the resulting new leaf vertex  with $\cL(a)\cup \cL(c)$ is called a {\it cherry reduction}. The new label $\cL(a)\cup \cL(c)$ is sometimes referred to as a {\it dummy taxon}.  We denote this reduction by $\cF[\{\cL(a),\cL(c)\}\rightarrow \cL(a)\cup \cL(c)]$. 
Reversely, we denote by $\cF[\cL(a)\cup \cL(c) \rightarrow \{\cL(a),\cL(c)\}]$  the operation of adjoining the vertex labeled $\cL(a)\cup \cL(c)$ with two new vertices labeled $\cL(a)$ and $\cL(c)$, respectively, via two new edges and deleting the label $\cL(a)\cup \cL(c)$.
{For an example of a cherry reduction, consider the two phylogenetic trees $S$ and $T$ of Figure~\ref{fig:forest}(i) that have a common cherry $\{f,g\}$. Reducing this cherry in $S$ and $T$ results in the two phylogenetic trees that are shown in Figure~\ref{fig:cherry}.}\\

We end this section with an important remark.

\begin{remark}
{The newly created leaf label, that results from applying a cherry reduction to a cherry $\{a,c\}$ that is common to two rooted phylogenetic trees, is the union of the labels associated with the vertices $a$ and $c$. For the rest of this paper, we therefore assume that the forest $\cF$ before applying a cherry reduction and the forest $\cF'$ that results from applying such a reduction have the same label set although the number of leaves has been decreased by one; thus $\cL(\cF)=\cL(\cF')$. Furthermore, we write $l(\cF)$ to denote the number of labeled vertices in $\cF$. Clearly, this number is always one greater than the number of leaves in $\cF$ due to the vertex labeled $\rho$. Lastly, let $S$ be a rooted binary phylogenetic tree. We write 
$l(S)=l(\cF)$ 
if the number of labeled vertices in $S$ and $\cF$ is identical and if there is a bijection between the vertex labels of $S$ and $\cF$.}
\end{remark}

\section{The algorithm {\sc allMAAFs}}\label{sec:alg}
In this section, we first give a brief outline of the algorithm {\sc allMAAFs} that calculates all maximum-acyclic-agreement forests for two rooted binary phylogenetic trees and, subsequently, present its pseudocode. Before doing so, we however start with an important remark to emphasize how the algorithm presented in this section separates itself from previously published work, and {give} some additional definitions.\\

\begin{remark}
 While  {\sc allMAAFs} has a similar flavor as an algorithm presented {in~\cite{whiddenWABI} that has been further improved in~\cite{whidden2010fast}, we remark here that our algorithm contains significant modifications due to a problem in both papers. In particular, Whidden et al.'s algorithms} are based on a different definition of an acyclic-agreement forest $\cF$ for  two rooted binary phylogenetic $\cX$-trees $S$ and $T$ compared to the definition that we have given in Section~\ref{sec:prelim}. Translated into the language of this paper, they define $\cF$ to be acyclic precisely if $AG(S,T,\cF)$ does not contain a directed cycle of length 2. Of course, this does not eliminate the possibility of cyclic inheritance in general although this is essentially required from a biological point of view. While {\sc allMAAFs} considers this stronger constraint and calculates a maximum-acyclic-agreement forest as defined in Section~\ref{sec:prelim}, we additionally show that our algorithm also computes all such forests (see Section~\ref{sec:proof}).
 \end{remark}

Let $S$ be a rooted binary phylogenetic $\cX$-tree, and let $\cF$ be a forest {such that $l(S)=l(\overline{\cF})$. Let $\{a,c\}$ be a cherry of {$S|_{\cL(\overline{\cF})}$}. We denote by $e_a$  the edge of $\cF$ that is incident with the leaf vertex, say $a'$, labeled $\cL(a)$, and by $e_c$  the edge of $\cF$ that is incident with the leaf vertex, say $c'$, labeled $\cL(c)$. 
Furthermore, if  $\{a,c\}$ is a 
contradicting cherry of $S$ and $\cF$ and $a \sim_{\cF} c$, let $F_i$ be the unique element of $\cF$ such that {$\cL(a)\subset\cL(F_i)$ and $\cL(c)\subset\cL(F_i)$}.
 Let $a',v_1,v_2,\ldots,v_n,c'$ be the path of vertices from $a'$ to $c'$ in $F_i$. We define $e_B=\{u,v\}$ to be an edge of $F_i$ such that $u\in\{v_1,v_2,\ldots,v_n\}$, $v\notin\{a',v_1,v_2,\ldots,v_n,c'\}$, and $u$ is an ancestor of $v$ in $F_i$. An example of an edge $e_B$ is shown in Figure~\ref{fig:alg}(i), where $e_1$ is such an edge for the contradicting cherry $\{a,b\}$ of the two topmost phylogenetic trees of that figure.} 
 Now, an edge $e$ of $\cF$ is said to be \emph{associated} with  a contradicting or common cherry $\{a,c\}$ for $S$ and $\cF$ if one of the following holds:
 \begin{enumerate}
 \item $e\in \{e_a, e_c\}$ if $\{a,c\}$ is a common cherry of $S$ and $\cF$, or $\{a,c\}$ is a contradicting cherry of $S$ and $\cF$ and $a \nsim_{\cF} c$, 
 \item  $e\in \{e_a, e_B,e_c\}$ if $\{a,c\}$ is a contradicting cherry of $S$ and $\cF$ and $a \sim_{\cF} c$.
 \end{enumerate}

We next describe the pseudocode of {\sc allMAAFs}. The
algorithm takes as input two rooted binary phylogenetic $\cX$-trees $S$ and $T$, a rooted binary phylogenetic tree $R$ and a forest $\cF$ such that $l(R)=l(\overline{\cF})$ and $\cL(T)=\cL(\cF)$, an integer $k$, and a list $M$ that contains information of previously reduced cherries. The output of {\sc allMAAFs} is a set $\boldsymbol{\cF}$ of forests for $\cF$ and an integer $k$. We will see in Section~\ref{sec:proof}, that if the input to {\sc allMAAFs} are two rooted binary phylogenetic $\cX$-trees $S$ and $T$, $R=S$, $\cF=T$, and $M=\emptyset$, then $\boldsymbol{\cF}$ precisely contains all maximum-acyclic-agreement forests for $S$ and $T$ and their respective hybridization number if and only if $k\geq h(S,T)$. We will therefore assume for the rest of the description of the pseudocode that {\sc allMAAFs}$(S,T,R, \cF,k,\emptyset)$ has initially been called for $R=S$, $\cF=T$, and $M=\emptyset$. 
If $k<0$, the algorithm immediately stops and returns an empty set. 
If, on the other hand, $k\geq 0$ and $l(R)=0$, 
then a forest $\cF'$ is obtained  from $\cF$ by calling {\sc cherryExpansion}$(\cF,M)$; that is undoing all previously performed cherry reduction. {As we will soon see in Lemma~\ref{l:af}}, $\cF'$ is an agreement forest for $S$ and $T$. Thus, if the graph $AG(S,T,\cF')$ is acyclic, then $\cF'$ is an acyclic-agreement-forest for $S$ and $T$, and the algorithm returns $\cF'$ and $|\cF'|-1$ with the latter being the hybridization number for $S$ and $T$ if $\cF$ is of smallest size.

Otherwise, if $k\geq 0$ and $l(R)>0$, the algorithm proceeds in a bounded-search type fashion by recursively deleting an edge in $\cF$ or reducing a {common cherry by calling {\sc cherryReduction}} until the resulting forest is a forest for $S$ and $T$. More precisely, each recursion starts by picking a cherry in $R$. Since $l(R)>0$, a cherry, say $\{a,c\}$, always exists since, by definition of $\overline{\cF}$, we have $l(R)\geq 2$. Depending on whether $\{a,c\}$ is a contradicting or common cherry of $R$ and $\cF$, and on whether or not $a$ and $c$ are vertices of the same component in $\cF$, the algorithm branches into at most three computational paths by recursively calling {\sc allMAAFs}. Note that the number of edge deletions that can additionally be performed at each step of the algorithm is given by the fifth parameter of each call to {\sc allMAAFs}. In the following, we say that a computational path {\it corresponds to deleting an edge in $\cF$} if  {\sc allMAAFs} is recursively called for a forest, say $\cF'$, that has been obtained from deleting an edge, and $R|_{\cL(\overline{\cF'})}$. Similarly, we say that a computational path {\it corresponds to calling {\sc cherryReduction}} if {\sc allMAAFs} is recursively called for a tree and a forest that are returned from a call to {\sc cherryReduction}.

Now, regardless of whether $\{a,c\}$ is a contradicting or common cherry of  $R$ and $\cF$, {\sc allMAAFs} branches into two new computational paths that correspond to deleting $e_a$ and $e_c$ in $\cF$, respectively. Additionally, if $\{a,c\}$ is a contradicting cherry of $R$ and $\cF$ and $a\sim_{\cF} c$, then {\sc allMAAFs} branches into a third computational path that corresponds to deleting an edge $e_B$ in $\cF$. Similarly, if $\{a,c\}$ is a common cherry of $R$ and $\cF$, then {\sc allMAAFs} branches into a third path that corresponds to calling {\sc cherryReduction}$(R,\cF,M,\{a,c\})$. {Intuitively, if $\{a,c\}$ is a contradicting cherry of $R$ and $\cF$, then, to obtain an agreement forest for the inputted trees $S$ and $T$, one needs to delete at least one of $e_a$, $e_c$ and $e_B$. Otherwise, if $\{a,c\}$ is a common cherry of $R$ and $\cF$, then, to obtain an acyclic-agreement forest, say $\cF'$ for $S$ and $T$, either the labels of $a$ and $c$ label vertices of the same component in $\cF'$, which is mimicked by calling {\sc cherryReduction} for $a$ and $c$, or the labels of  $a$ and $c$ are contained in the label sets of two distinct elements in $\cF'$; thus one needs to delete one of $e_a$ or $e_c$. Noting that a common cherry of $R$ and $\cF$ is not necessarily a common cherry of $S$ and $T$, we remark that this part of the algorithm has a  similar flavor as~\cite[Lemma 3.1.2]{sempbordfpt2007}, where the authors consider so-called common chains of $S$ and $T$ with at least 3 leaves.}
To keep track of the set $\boldsymbol{\cF}$ of maximum-acyclic-agreement forests for the two initially inputted rooted binary phylogenetic $\cX$-trees $S$ and $T$, the algorithm uses the variable $k_{\min}$ which is set to be the minimum number of edges that have so far been deleted in $\cF$ so that the resulting forest is an acyclic-agreement forest for $R$ and $\cF$, and updates $\boldsymbol{\cF}$ and $k_{\min}$ as appropriate throughout a run of {\sc allMAAFs}.

We end the description of the pseudocode by noting that {\sc allMAAFs} always terminates because, at each recursive call, either $k$ is decreased by one or the number of leaves in $R$ is decreased by one due to calling {\sc cherryReduction}.

\begin{small}
\begin{algorithm}[]
\KwData{{A rooted binary phylogenetic tree $R$ and a forest $\cF$ such that $l(R)=l(\overline{\cF})$, a list $M$ that contains all information of previously applied cherry reductions, and a common cherry $\{a,c\}$ of $R$ and $\cF$.}}
\KwResult{A rooted binary phylogenetic tree $R'$ and a forest $\cF'$ obtained from $R$ and $\cF$, respectively by replacing $\{a,c\}$ with  a single leaf with a new label $\cL(a)\cup \cL(c)$,  and an updated list $M'$.}

$M'\gets \mbox{Add } \{\cL(a), \cL(c)\}\mbox{ as last element of } M$\;
 $R' \gets R[\{\cL(a),\cL(c)\}\rightarrow \cL(a)\cup \cL(c)]$\;
$\cF'\gets\cF[\{\cL(a),\cL(c)\}\rightarrow \cL(a)\cup \cL(c) ]$\;

\Return $(R',\cF',M')$
\caption{ \textsc{cherryReduction}$(R,\cF,M,\{a,c\})$\label{a:subRed}} 
\end{algorithm}
\end{small}

\begin{small}
\begin{algorithm}[]
\KwData{A  forest $\cF$ and a list $M$ containing information of all previously applied cherry reductions.}
\KwResult{A forest $\cF$ whose vertices labeled with dummy taxa have been replaced by the corresponding cherries using the information contained in $M$.}
\While{$M$ {\em is not empty}}{ 
         $M\gets \mbox{remove the last element, say $\{\cL(a), \cL(c)\}$, from } M$\;
	$\cF\gets \cF[\cL(a)\cup \cL(c) \rightarrow \{\cL(a),\cL(c)\}]$\;
}
\Return $\cF$
\caption{ \textsc{cherryExpansion}$(\cF,M)$\label{a:subExp}} 
\end{algorithm}
\end{small}

\begin{small}
\begin{algorithm}[]
\KwData{{Two rooted binary phylogenetic $\cX$-trees $S$ and $T$, a rooted binary phylogenetic tree $R$ and a forest $\cF$ such that $l(R)=l(\overline{\cF})$ and $\cL(T)=\cL(\cF)$, an integer $k$, and a list $M$ that contains information of previously reduced cherries.}
}
\KwResult{A set $\boldsymbol{\cF}$ of forests for $\cF$ and an integer. In particular, if $\cF=T$, $R=S$, $M=\emptyset$, and $k\geq h(S,T)$ is the input to {\sc allMAAFs}, the output precisely consists of all maximum-acyclic-agreement forests for $S$ and $T$ and their respective hybridization number.}

     \If{$k<0$}{\Return $(\emptyset,k-1)$\;}
     \If{$|\cL(R)|=0$}{
     		    $\cF' \gets$ \textsc{cherryExpansion}($\cF$, $M$)\;
		    \If {$AG(S,T,\cF')$ \textnormal{is acyclic}}{
       			 \Return ($\cF'$, $|\cF'|-1$)\;
			 }
			 \Else{\Return $(\emptyset,k-1)$\;}
	}
	\Else{
		let $\{a,c\}$ be a  cherry of $R$\;
		 ($\boldsymbol{\cF_a}$, $k_a$) $\gets$ \textsc{allMAAFs}($S$, $T$, $R|_{\cL(\overline{\cF-\{e_a\}})}$, $\cF-\{e_a\}$, $k-1$, $M$)\;
				 \If{$\boldsymbol{\cF_a}\neq \emptyset$}{$k' \gets \min(k, k_a+1 - |\cF|)$\;}
				 ($\boldsymbol{\cF_c}$, $k_c$) $\gets$ \textsc{allMAAFs}($S$, $T$, $R|_{\cL(\overline{\cF-\{e_c\}})}$, $\cF-\{e_c\}$, $k'-1$, $M$)\;
				  \If{$\boldsymbol{\cF_c}\neq \emptyset$}{$k' \gets \min(k',k_c+1 - |\cF|)$\;}
				 $\boldsymbol{\cF} \gets \emptyset$\;
				 
		\If{$\{a,c\}$ \textnormal{ is a contradicting cherry of } $R$ \textnormal{ and } $\cF$}{
			\If{$a \nsim_{\cF} c$ }{
				 
				 $k_{\min} \gets k'$\;
				 \lIf{$(k_a +1- |\cF|= k_{\min})$}{
				 	$\boldsymbol{\cF} =  \boldsymbol{\cF_a}$\;
				}
				\lIf{$(k_c +1- |\cF|= k_{\min})$}{
				 	$\boldsymbol{\cF} = \boldsymbol{\cF} \cup \boldsymbol{\cF_c}$\;
				}
				 \Return ($\boldsymbol{\cF}$, $k_{\min}-1$)\;
			}
			 \Else{
				  	($\boldsymbol{\cF_B}$, $k_B$) $\gets$ \textsc{allMAAFs}($S$, $T$, $R|_{\cL(\overline{\cF-\{e_B\}})}$, $\cF-\{e_B\}$, $k'-1$, $M$)\;
				  \If{$\boldsymbol{\cF_B}\neq \emptyset$}{$k_{\min} \gets \min(k', k_B+1 - |\cF| )$\;} 
				 \lIf{$(k_a +1 - |\cF| = k_{\min})$}{
				 	$\boldsymbol{\cF} =  \boldsymbol{\cF_a}$\;
				}
				 \lIf{$(k_B +1 - |\cF| = k_{\min})$}{
				 	$\boldsymbol{\cF} = \boldsymbol{\cF} \cup \boldsymbol{\cF_B}$\;
				}
				\lIf{$(k_c +1 - |\cF| = k_{\min})$}{
				 	$\boldsymbol{\cF} = \boldsymbol{\cF} \cup \boldsymbol{\cF_c}$\;
				}
				 \Return ($\boldsymbol{\cF}$, $k_{\min}-1$)
			 }
		}
	\Else{
				 $(R',\cF',M') \gets$ \textsc{cherryReduction}($R$, $\cF$, $M$, $\{a,c\}$)\;
				  ($\boldsymbol{\cF_r}$, $k_r$) $\gets$ \textsc{allMAAFs}($S$, $T$, $R'|_{\cL(\overline{\cF'})}$, $\cF'$, $k'$, $M'$)\;
				   \If{$\boldsymbol{\cF_r}\neq \emptyset$}{$k_{\min} \gets \min(k', k_r +1 - |\cF|)$\;}

				 \lIf{$(k_a +1 - |\cF| = k_{\min})$}{
				 	$\boldsymbol{\cF} =  \boldsymbol{\cF_a}$\;
				}
				\lIf{$(k_c +1 - |\cF| = k_{\min})$}{
				 	$\boldsymbol{\cF} = \boldsymbol{\cF} \cup \boldsymbol{\cF_c}$\;
				}
				\lIf{$(k_r +1 - |\cF| = k_{\min})$}{
				 	$\boldsymbol{\cF} = \boldsymbol{\cF} \cup \boldsymbol{\cF_r}$\;
				}
				 \Return ($\boldsymbol{\cF}$, $k_{\min}-1$)\;
	}
	}
\caption{ \textsc{allMAAFs}($S$, $T$, $R$, $\cF$, $k$, $M$)\label{a:recForests}} 
\end{algorithm}
\end{small}

\begin{figure}
\begin{center}
\begin{tabular}{c}
\includegraphics[width = 7cm]{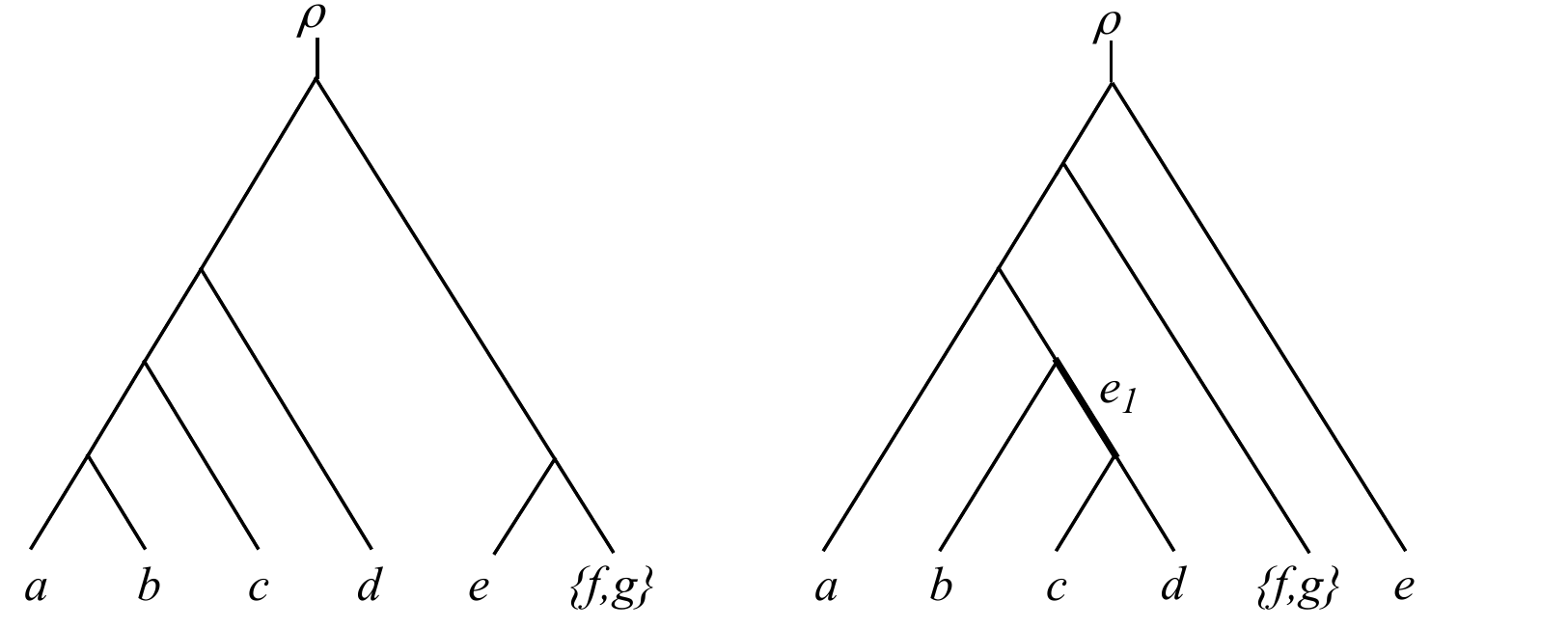}\\
($i$)\\
\includegraphics[width = 7cm]{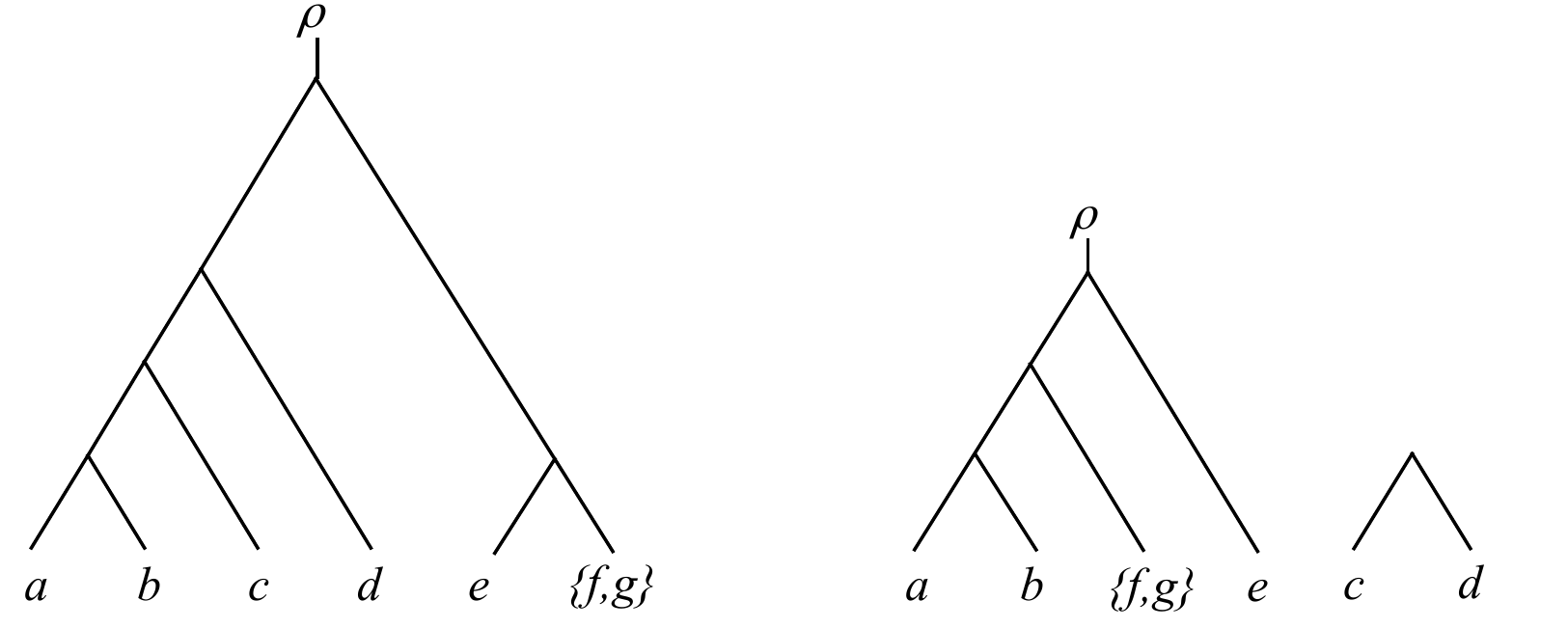}\\
($ii$)\\
\includegraphics[width = 7cm]{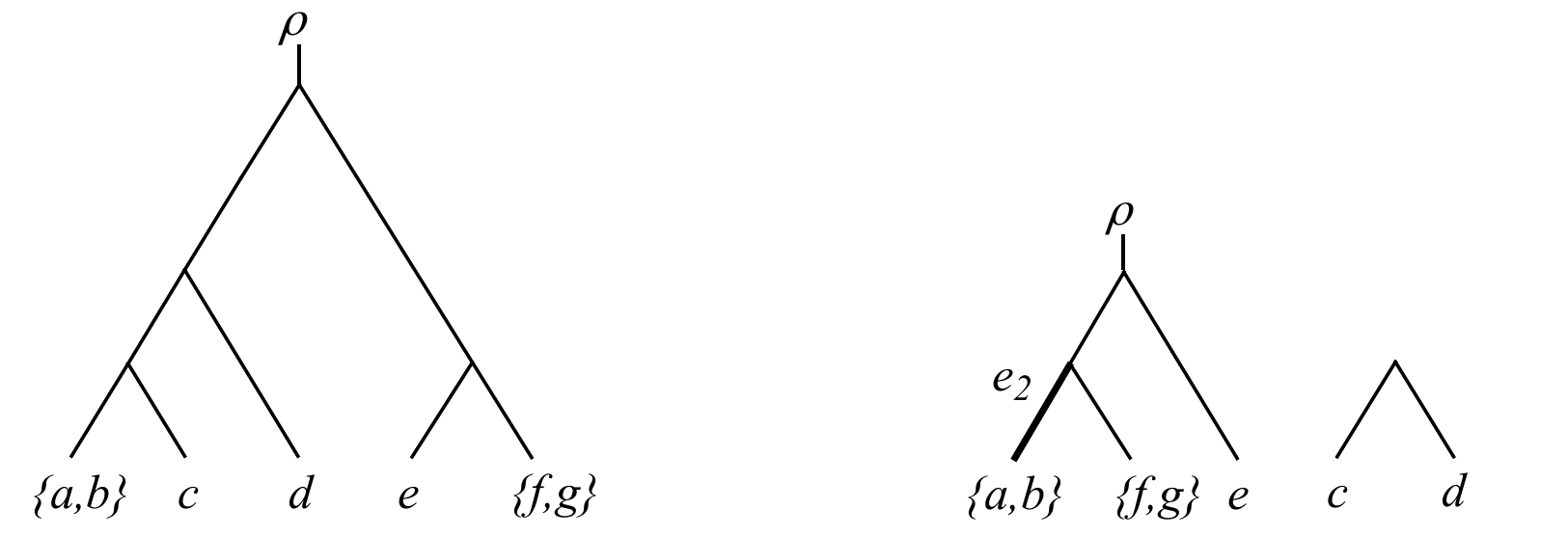}\\
($iii$)\\
\includegraphics[width = 7cm]{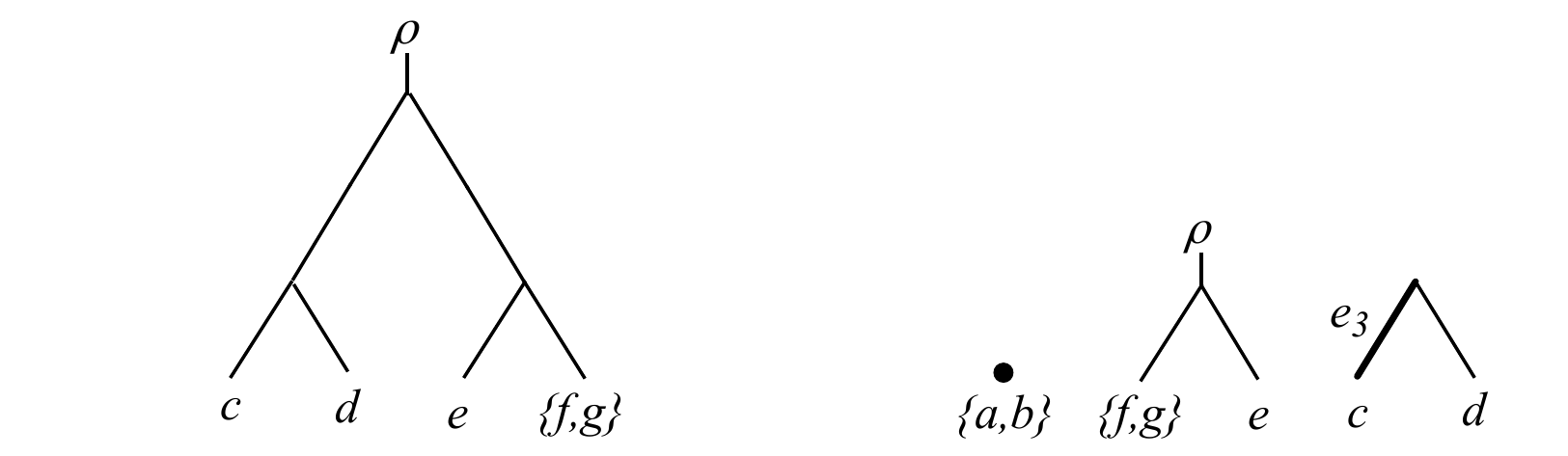}\\
($iv$)\\
\includegraphics[width = 7cm]{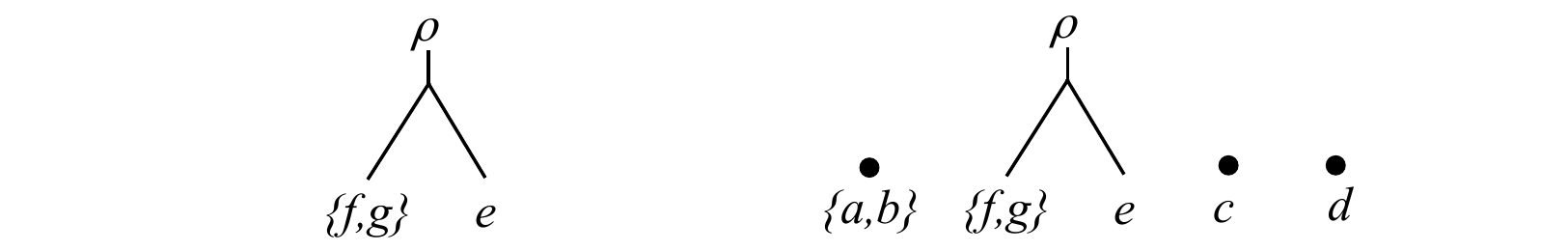}\\
($v$)\\
\includegraphics[width = 7cm]{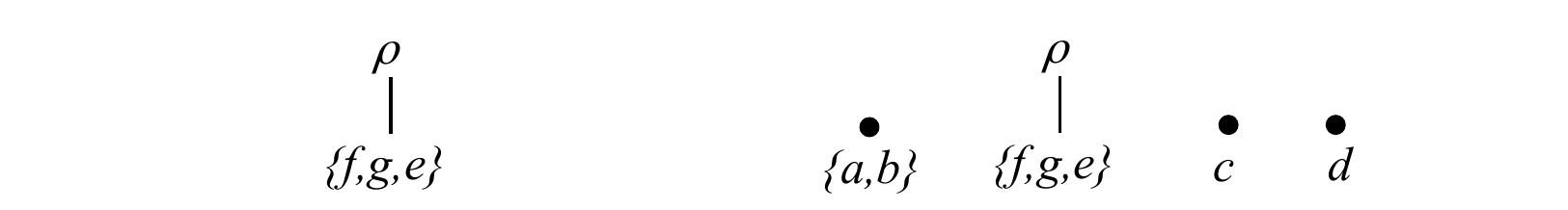}\\
($vi$)\\
\includegraphics[width = 3.5cm]{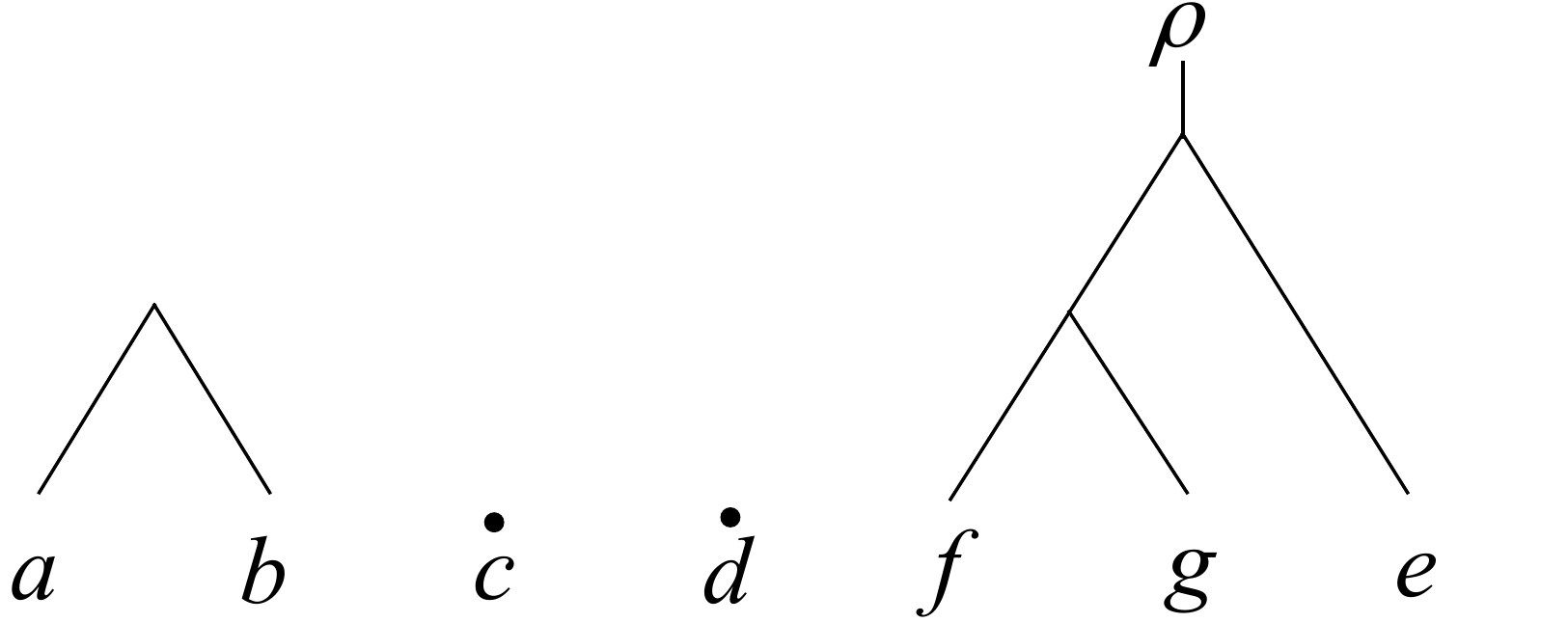}\\
($vii$)\\
\end{tabular}
\caption{An example of a call to {\sc processCherries}$(S, T,\fc_1,\fc_2,\ldots,\fc_6)$, where $S$ and $T$ are the phylogenetic trees of Figure \ref{fig:forest}(i) and the cherry list is $((\{f,g\}, \emptyset)$, 
$(\{a,b\}, e_1)$, 
$(\{a,b\}, \emptyset )$, 
$(\{\{a,b\},c\}, e_2 )$, 
$(\{c,d\}, e_3 )$, 
$(\{\{f,g\},e\}, \emptyset))$. In (i)-(vi), the phylogenetic trees and the forests are shown that are obtained by {successively} analyzing each cherry action of the above list while in (vi) 
the result of the call
{\sc cherryExpansion($\cF$, $M$)} is shown, where $\cF$ is the  forest that is depicted in (vi) and $M$ contains all the information of previously applied cherry reductions. Note that the forest in (vii) is a maximum-acyclic-agreement forest for $S$ and $T$. 
}
\label{fig:alg}
\end{center}

\end{figure}
\section{Correctness of the algorithm {\sc allMAAFs}\label{sec:proof}}
In this section, we prove the main result of this paper. In particular, we show that the algorithm {\sc allMAAFs} calculates all maximum-acyclic-agreement forests for two rooted binary phylogenetic trees $S$ and $T$ for when inputted with $R=S$, $\cF=T$, {$M=\emptyset$}, and $k \geq h(S,T)$. We start with some additional definitions that are crucial for what follows. \\


Let $\cF$ and $\cG$ be two forests such that $\cL(\cF)=\cL(\cG)$. We call $\cG$ a \emph{super-forest} of $\cF$ if and only if the following two conditions are satisfied:
\begin{enumerate}
\item [(1)]for each $G_j \in \cG$, there exists a subset $\cF'$ of $\cF$ such that $\cL(\cF')=\cL(G_j)$, and
\item [(2)] for each leaf vertex $a$ in an element of $\cG$, there exists a component $F_i$ in $\cF$ such that $\cL(F_i)\supseteq \cL(a)$.
\end{enumerate}
\noindent {For an example of two forests that are no super-forests of the forest that is shown in Figure~\ref{fig:forest}(ii), see Figure~\ref{fig:super}.}

\begin{figure}[t]
\begin{center}
\begin{tabular}{cc}
\includegraphics[width = 5cm]{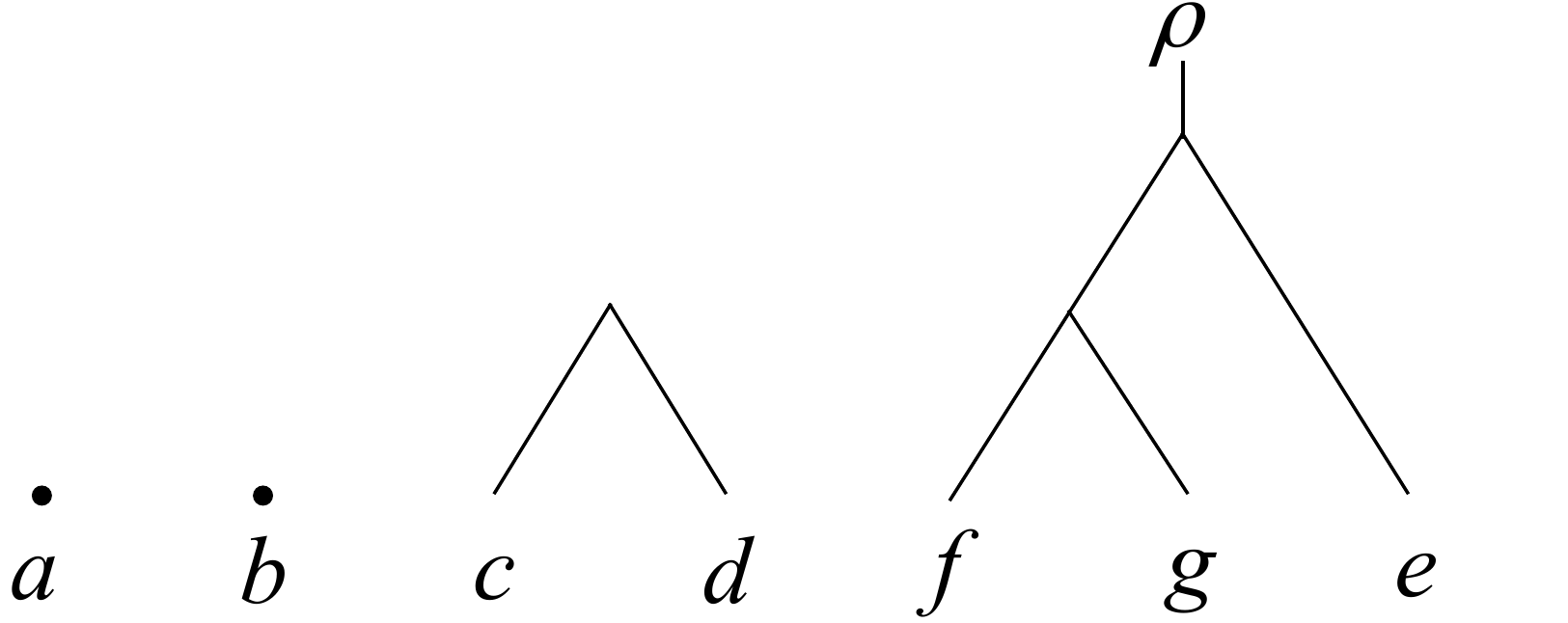}
&
\includegraphics[width = 4.0cm]{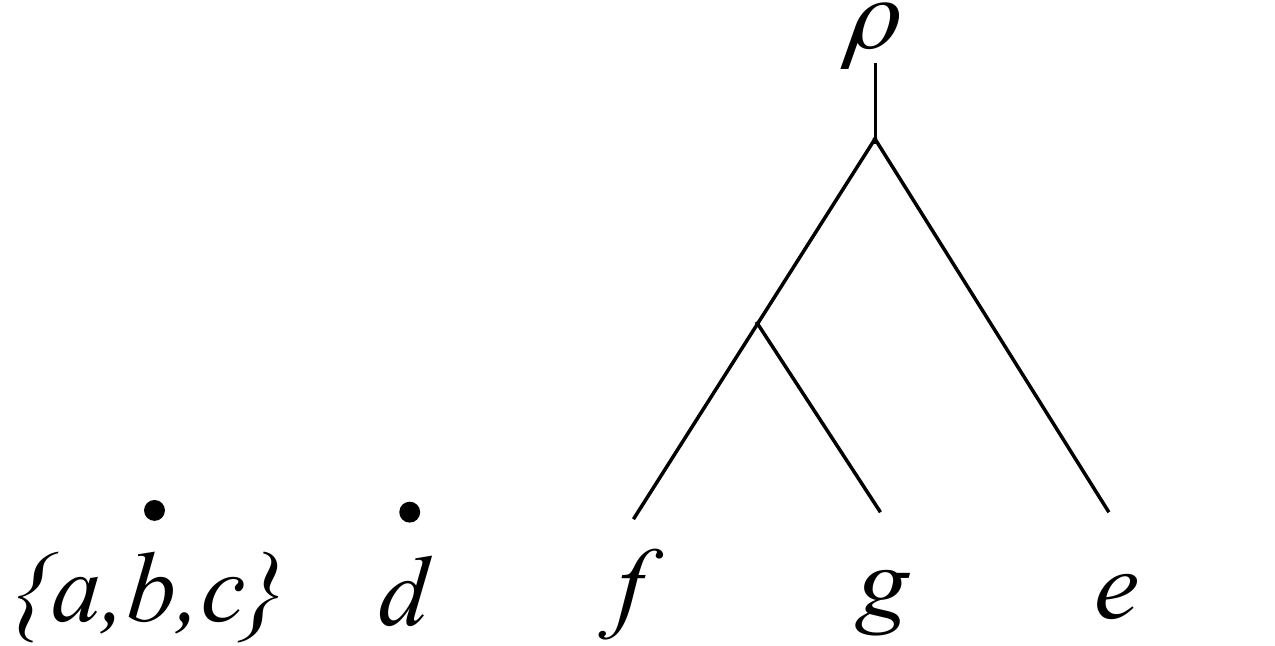}
\\
\end{tabular}
\caption{Two non-super-forests of the forest $\cF$ in Figure \ref{fig:forest}($ii$). The forest on the left-hand side is not a super-forest of $\cF$ because there exists no subset $\cF'$ of $\cF$ such that $\cL(\cF')=\{a\}$.  The forest on the right-hand side  is not a super-forest of $\cF$ because there exists no component $F_i$ in $\cF$ such that $ \cL(F_i) \supseteq \{a,b,c\}$. }
\label{fig:super}
\end{center}
\end{figure}

The next observation is an immediate consequence of the previous definition.

\begin{observation}\label{obs:baseCase}
Given an acyclic-agreement forest $\cF$ for two rooted binary phylogenetic $\cX$-trees $S$ and $T$, then {$S$ and $T$} are both super-forests for $\cF$. 
\end{observation}

Let $R$ be a rooted binary phylogenetic tree and let $\cF$ be a forest such that $l(R)=l(\overline{\cF})$. Furthermore, let $\{a,c\}$ be a cherry of $R$. In the following, we say that a pair $\fc=( \{a,c\}, e)$ is a \emph{cherry action} if one of the following conditions is satisfied:
\begin{enumerate}
\item [(1)] $\{a,c\}$ is a contradicting or common cherry of $R$ and $\cF$ and $e$ is an edge associated with $\{a,c\}$, or
\item [(2)] $\{a,c\}$ is a common cherry of $R$ and $\cF$ and $e=\emptyset$.
\end{enumerate}
Finally, we say that $\fcL=(\fc_1,\fc_2, \ldots \fc_l)$ is a {\em cherry list} for $R$ and $\cF$ if and only if each $\fc_i$ is a cherry action in iteration $i$ of the following algorithm; i.e. {\sc processCherries} does not return {\it false}:

\begin{tabbing}
{\sc processCherries}($R,\cF,(\fc_1,\fc_2,\ldots,\fc_l)$)\\
\quad~ $M\gets \emptyset$; \\
\quad~ \={\bf for} \= {\bf each} $i=1,\ldots, l$\\
 \>   \>  $(\{a,c\},e_i ) \gets \fc_i$;\\
 \>  \> {\bf if} $\{a,c\}$ {is a common  cherry} of $R$ and $\cF$ and $e_i=\emptyset$ \\
 \> \> \quad $(R,\cF,M) \gets$ \textsc{cherryReduction}($R$, $\cF$, $M$, $\{a,c\}$);\\
 \>  \>{\bf else if } $\{a,c\}$ is a {common or contradicting cherry} of $R$ and $\cF$ and $e_i$ is associated with $\{a,c\}$\\
\> \> \quad $\cF \gets \cF - \{e_i\}$;\\
\> \> \quad $R \gets R|_{\cL(\overline{\cF})}$;\\
 \>  \> {\bf else}\\
 \> \> \quad \Return({\it false});\\
 \quad~ $\cF\gets$ {\sc cherryExpansion}$(\cF,M)$;\\
\quad~ \Return $(R,\cF,M)$;
\end{tabbing}

\begin{remark}
The algorithm {\sc processCherries} is mimicking a computational path of the algorithm {\sc allMAAFs} for when the former algorithm is given a cherry list for $R$ and $\cF$. A specific example of a call to {\sc processCherries} is shown in Figure~\ref{fig:alg} with a detailed description given in the caption of this figure.
\end{remark}

\begin{sloppypar}
In what follows, we will sometimes make use of the algorithm {\sc processCherries}$(R,\cF,\fcL)$, but without executing the call to {\sc cherryExpansion} in the second-to-last line of this algorithm.
We refer to this slightly different algorithm as {\sc processCherries}*$(R,\cF,\fcL)$ and to the returned forest as a {\it reduced forest}. 
Now, let $\cF'$ be the forest obtained from calling {\sc processCherries}$(R,\cF,\fcL)$, and let $\cF''$ be the forest obtained from calling {\sc processCherries*}$(R,\cF,\fcL)$. We say that $\cF'$ is the {\it underlying forest for $\cF''$} and observe that $|\cF'|=|\cF''|$. 

We continue with two important remarks.
\end{sloppypar}

\begin{remark}\label{rem:onlyOneLeafF}
Applying {\sc processCherries}*$(R,\cF,\fcL)$, returns a tree $R$ that does not contain any vertex if and only if, prior to calling {\sc cherryExpansion}$(\cF,M)$, the forest $\cF$ only consists of isolated vertices and possibly an element that precisely contains a vertex labeled $\rho$ that is attached to a vertex by an edge; i.e. $\overline{\cF}=\emptyset$ (for an example, see Figure \ref{fig:alg}(vi)).
\end{remark}

\begin{remark}\label{rem:noOneLeaf}
By the definition of $\overline{\cF}$, note that applying {\sc processCherries}* never returns a tree $R$ that consists of a single leaf attached to the root vertex labeled $\rho$.
\end{remark}


Now, let $\cG$ and $\cF$ be two forests such that $\cG$ is a super-forest of $\cF$.
Furthermore, let $e$ be an edge and $\{a,c\}$ be a cherry (if it exists) of $\cG$. We say that $e$ {\it is a bad choice} for $\cG$ and $\cF$ if $\cG-\{e\}$ is not a super-forest of $\cF$. Note that $\cG-\{e\}$ always satisfies Condition (2) in the definition of a super-forest. Similarly, we say that $\{a,c\}$ {\it is a bad choice} for $\cG$ and $\cF$ if the forest, say $\cG'$,  that that is obtained from $\cG$ by reducing the cherry $\{a,c\}$ to a new leaf is not a super-forest of $\cF$. Note that $\cG'$ always satisfies Condition (1) in the definition of  a super-forest.\\



We next prove two lemmas that are necessary to establish the main result (Theorem~\ref{thm:main}) of this paper.

\begin{lemma}\label{lemma:cuttingCorrectly}
Let $\fcL$ be a cherry list for two rooted binary phylogenetic $\cX$-trees $S$ and $T$, and let $\cF$ be an maximum-acyclic-agreement forest for  $S$ and $T$. Additionally, let $S'$ and $\cG'$ be  the  tree and the  forest, respectively, that have been obtained from calling  \mbox{{\sc processCherries}*$(S,T,\fcL)$}, and let $\cG$ be the underlying forest for $\cG'$. If $\cG'$ is a super-forest for $\cF$, then $S'$  contains at least one cherry or $\cG=\cF$. 
\end{lemma}
\begin{proof}
Suppose that this is not true. 
 Thus,  $l(S')=0$ (see Remark \ref{rem:noOneLeaf})  and there exists an element in $\cF$ that is not an element in $\cG$. 
Furthermore, since both forests are forests of $T$, we cannot have that there exists $F_i\in \cF$ and $G_j \in \cG$ such that $\cL(F_i ) =  \cL(G_j)$ and $F_i \not\cong G_j$. 
 Since $\cG'$ is a super-forest for $\cF$ {other than $\cF$}, there  exist at least two components $F_i$ and $F_j$ of $\cF$ such that $\cL(F_i)\cup\cL(F_j)\subseteq\cL(G_k)$, where $G_k$ is an element of $\cG$. 
 Furthermore, since 
 $l(S')=0$, we have $\overline{\cG'}=\emptyset$ (see Remark \ref{rem:onlyOneLeafF}).
 Now, since $G_k\in \cG$, either $G_k$ is an isolated vertex $a$ 
 such that $\cL(a)=\cL(G_k) \supseteq (\cL(F_i) \cup \cL(F_j))$, or $G_k$ is a leaf vertex that is attached to the vertex labeled $\rho$ by an edge such that $\cL(a) \cup \{\rho\} = \cL(G_k) \supseteq (\cL(F_i) \cup \cL(F_j))$. Since neither $\cL(F_i)=\{\rho\}$ nor $\cL(F_j)=\{\rho\}$  (see~\cite[Lemma 1]{BaroniEtAl2005}), $\cG'$ does not fulfill Condition (2) in the definition of a super-forest; a contradiction.
\end{proof}

\begin{lemma}\label{l:af}
Let $S$ and $T$ be two rooted binary phylogenetic $\cX$-trees, and let $\cF$ be a forest that is returned from calling {\sc cherryExpansion} (line 4 of the pseudocode of Algorithm~\ref{a:recForests}) while executing ${\textnormal{\sc allMAAFs}}(S,T,S,T,k,\emptyset)$. Then, $\cF$ is an agreement forest for $S$ and $T$.
\end{lemma}

\begin{proof}
Let $\cF'$ be the forest for which calling {\sc cherryExpansion} returns $\cF$. By construction, $\cF$ is a forest for $T$. Now, for the purpose of deriving a contradiction, assume that $\cF$ is not a forest for $S$. 
Since $\cF$ is a forest for $T$ and $\cL(S)=\cL(T)$, the label sets of the elements in $\cF$ partition $\cL(S)$. Thus, it is sufficient to consider the following two cases:\\

\noindent{\bf Case (1).} Assume that there exists an element $F_i$ in $\cF$ such that $S|_{\cL(F_i)}\ncong F_i$. Since $l(R)=0$ (line 3 of the pseudocode of Algorithm~\ref{a:recForests}), we have by Remark \ref{rem:onlyOneLeafF} that $\overline{\cF'}=\emptyset$.
This implies that the element of $\cF$ with leaf sets $\cL(F_i)$ has been shrunk to a single vertex or to a single leaf that is attached to the vertex labeled $\rho$ by an edge. But, since $S|_{\cL(F_i)}\ncong F_i$, one of the cherry reductions that has been used to shrink $T|_{\cL(F_i)}$ is called for a cherry that is not a cherry of $R$, where $R$ is the tree that is considered in some recursive call of  ${\textnormal{\sc allMAAFs}}(S,T,S,T,k,\emptyset)$ (see pseudocode of Algorithm~\ref{a:recForests}); a contradiction.\\


\noindent{\bf Case (2).} Assume that there exist two elements $F_i$ and $F_j$ in $\cF$ such that $S(\cL(F_i))$ and $S(\cL(F_j))$ are not vertex-disjoint in $S$. 
Let $S'=S|_{\cL(F_i) \cup \cL(F_j)}$. For example, the simplest case is shown in Figure \ref{fig:helpLemma2}, where the subtrees in white are part of $F_i$ and the ones in black of $F_j$. In general, a straightforward check now shows  that it is not possible to shrink both $T|_{\cL(F_i)}$ and  $T|_{\cL(F_j)}$ to two distinct single vertices in $\cF'$ (one possibly being attached to the vertex labeled $\rho$) by using cherry reductions because to shrink one of the two components to a single vertex it is necessary to cut a subtree of the other component, thereby contradicting that $F_i$ and $F_j$ are both elements in $\cF$. Referring back to Figure \ref{fig:helpLemma2}, $F_j$ cannot be shrunk to a single vertex by using a list of cherry reductions without cutting $S_1$.
\\
\\
Combining both cases establishes the lemma.
\end{proof}

\begin{figure}
\begin{center}
\begin{tabular}{c}
\includegraphics[width = 4.5cm]{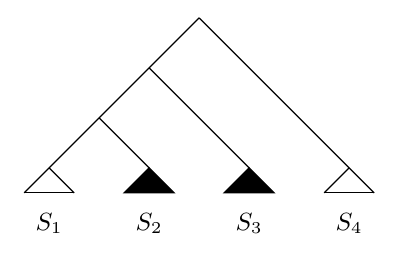}
\\
\end{tabular}
\caption{An example of a rooted phylogenetic tree $S'$ that is  used in Case (2) of the proof of Lemma~\ref{l:af} (for details, see text).}
\label{fig:helpLemma2}
\end{center}
\end{figure}

\begin{theorem}\label{thm:main}
Let $S$ and $T$ be two rooted binary phylogenetic $\cX$-trees. Calling $$\textnormal{{\sc allMAAFs}}(S,T,S,T,k,\emptyset)$$ returns all maximum-acyclic-agreement forests for $S$ and $T$ if and only if $k \geq h(S,T)$. 
\end{theorem}

\begin{proof}
{By Lemma~\ref{l:af}, each forest that is calculated in the course of executing {\sc allMAAFs}$(S,T,S,T,k,\emptyset)$ and checked for acyclicity (see line 5 of the pseudocode of Algorithm~\ref{a:recForests}) is an  agreement forest for $S$ and $T$. Thus, if $k\geq h(S,T)$, each forest that is returned from running the algorithm is an acyclic-agreement forest for $S$ and $T$. Moreover, 
since $k$ is updated to take advantage of the size of the best solutions that
previous recursive calls have found (lines 13, 16, 27 and 36), only maximum-acyclic-agreement forests are returned. It is therefore sufficient to show that each maximum-acyclic-agreement forest for $S$ and $T$ is returned by the algorithm.} 

Let {\sc allMAAFs}$(S,T,S,T,k,\emptyset)$ be a call of Algorithm~\ref{a:recForests}, and, {for each $l\in\{1,2,\ldots,h(S,T)+1\}$}, let $\cG_{l}$ be the set of all reduced forests of size $l$ that have been computed by executing this call. In other words, $\cG_l$ precisely  contains all forests that are used as a parameter in a recursive call to {\sc allMAAFs} in lines 11, 14, 25  and 34 of the pseudocode and, in particular, $T\in \cG_1$. Furthermore, let $\cF$ be a maximum-acyclic-agreement forest for $S$ and $T$. We will prove that, for each $l\in\{1,2,\ldots,h(S,T)+1\}$, the set
$\cG_{l}$ contains a reduced forest $\cG'$ that is a super-forest for $\cF$. 
This implies that $\cG_{h(S,T)+1}$ contains a reduced forest $\cG'$ that is a super-forest of $\cF$ such that $|\cF|=|\cG|$, where $\cG$ is the underlying forest of $\cG'$. Hence, as $\cG$ and $\cF$ are both forests for $T$, it follows that $\cG$ is isomorphic to $\cF$, thereby establishing the theorem. 

We proceed by induction on $l$. If $l=1$, then the result follows from Observation \ref{obs:baseCase} and because  $T\in\cG_1$.  Now suppose that the result holds whenever $l \leq h(S,T)$. We will next show that the claim holds for $l +1$. Let  $\cG'$ be a reduced forest $\cG'$ of size $l$ such that $\cG'$ is a super-forest for $\cF$. By the induction assumption, $\cG'$ exists. Let $\cG$ be the underlying forest of $\cG'$. Furthermore, let $\fcL^{\cG}$ be the cherry list  that has been used by calling  {\sc allMAAFs}$(S,T,S,T,k,\emptyset)$ to construct $\cG'$, and let $R$ be the phylogenetic tree that is returned from calling {\sc processCherries*}($S,T,\fcL^{\cG}$). Since $|\cG|<|\cF|$, it follows from  Lemma  \ref{lemma:cuttingCorrectly} that $R$ contains a cherry $\{a,c\}$. Let $\cL(a) \subset \cX$  and $\cL(c)  \subset \cX$ be the label sets of the leaf vertices $a$ and $c$, respectively, in $R$, and let $a'\in\cL(a)$ and $c'\in\cL(c)$. Furthermore, if $\{a,c\}$ is a contradicting cherry for $R$ and $\cG'$ and $a \sim_{\cG'} c$, let $\cL(B)\subset\cX$ be the union of labels of all leaf vertices that are contained in the pendant subtree below $e_B$ in $\cG'$. 
Note that, since $\cG'$ is a super-forest for $\cF$, we have that there exist two elements $F_i, F_j \in \cF$, not necessarily distinct, such that $\cL(a) \subseteq \cL(F_i)$ and  $\cL(c) \subseteq \cL(F_j)$. The rest of the proof distinguishes two cases depending on whether $\{a,c\}$ is a contradicting or common cherry for $R$ and $\cG'$.

First, suppose that $\{a,c\}$ is a contradicting cherry for $R$ and $\cG'$. To derive a contradiction, assume that $\cG_{l+1}$ does not contain any reduced forest that is a super-forest of $\cF$. 
In particular, this implies that deleting any edge associated with $\{a,c\}$ is a bad choice
for $\cG'$ and $\cF$ 
since no resulting forests is a super-forests for $\cF$ although they all satisfy Condition (2) in the definition of a super-forest.  Thus, one of the following holds:
\begin{enumerate}
\item[(1)] $a \nsim_{\cG'} c$ and both edges $\{e_a\}$ and  $\{e_c\}$ are bad choices for $\cG'$ and $\cF$;
\item[(2)] $a \sim_{\cG'} c$ and all edges $\{e_a\}$, $\{e_B\}$ and $\{e_c\}$ are bad choices for $\cG'$ and $\cF$.
\end{enumerate}


\noindent {\bf Case (1).} Observe that neither $\cG'-\{e_a\}$ nor $\cG'-\{e_c\}$ is a super-forest of $\cF$. 
This implies that $\cF$ does not contain an element $F_i$
such that $\cL(a) =\cL(F_i)$ or $\cL(c) =\cL(F_i)$. Thus $\cF$ contains two distinct components $F_j$ and $F_k$  such that $\cL(a)\subset\cL(F_j)$ and $\cL(c)\subset\cL(F_k)$, and for which  there exist elements $x,y \in \cX$ such that $x\in \cL(F_j)$, $x\notin \cL(a)$, $y\in \cL(F_k)$, and $y\notin \cL(c)$. By construction, each of $x$ and $y$ is contained in a label of a distinct leaf in $R$. Now, recalling that $\{a,c\}$ is a cherry of $R$, we have that $\mrca_R(a',c',x,y)$ is an ancestor of $\mrca_R(a',c')$ and, therefore, $\mrca_S(a',c',x,y)$ is an ancestor of $\mrca_S(a',c')$.
Furthermore, as $a',x \in \cL(F_j)$ and  $c',y \in \cL(F_k)$,  it now follows that $S(\cL(F_j))$ and $S(\cL(F_k))$ do both have the vertex $\mrca_S(a',c')$ in common; thereby contradicting that $\cF$ is an agreement forest for $S$ and $T$.\\

\noindent{\bf Case (2).} Observe that no forest in $\{\cG'-\{e_a\},\cG'-\{e_B\},\cG'-\{e_c\}$ is a super-forests of $\cF$.  
This implies that  $\cF$ does not contain any 
element $F_i$ such that $\cL(a) =\cL(F_i)$ or $\cL(c) =\cL(F_i)$ or a subset $\cF'$ of $\cF$ such that $\cL(B) =\cL(\cF')$. We next consider three subcases. 

First, assume that $\cF$ contains a component $F_j$ such that $\cL(a)\subset \cL(F_j)$, $\cL(c)\subset \cL(F_j)$ and there exists at least one element in the intersection $\cL(B) \cap \cL(F_j)$. Let $b'$ be such an element.
By construction, each of $a'$, $b'$, and $c'$ is contained in a label of a distinct leaf in $R$. Now, recalling that $\{a,c\}$ is a cherry of $R$, we have that $\mrca_R(a',b',c')$ is an ancestor of $\mrca_R(a',c')$ and, therefore, $\mrca_S(a',b',c')$ is an ancestor of $\mrca_S(a',c')$. On the contrary, let $G_k$ be the element of $\cG'$ that contains the leaf labeled $\cL(a)$ and the leaf labeled $\cL(c)$. Since $G_k$ also contains a leaf whose label contains $b'$ and due to the definition of $e_B$, {we have that $\mrca_{G_k}(a',b',c')$ is an ancestor of $\mrca_{G_k}(a',b')$ or $\mrca_{G_k}(c',b')$ and, therefore, $\mrca_{T}(a',b',c')$ is an ancestor of $\mrca_{T}(a',b')$ or $\mrca_{T}(c',b')$}. Thus, $S|\{a',b',c'\}\ncong T|\{a',b',c'\}$; thereby contradicting that $\cF$ is a maximum-acyclic-agreement forest for $S$ and $T$.

Second, assume that $\cF$ contains a component $F_j$ such that $\cL(a)\subset \cL(F_j)$, $\cL(c)\subset \cL(F_j)$, and $\cL(B) \cap \cL(F_j) = \emptyset$. Then, since there exists no subset $\cF'$ of $\cF$ such that $\cL(B) =\cL(\cF')$, there exists a distinct element $F_k\in\cF$ such that $b' \in \cL(F_k)$ for any $b'\in\cL(B)$ and there exists an element $x\in\cX$ for which $x\in \cL(F_k)$ and $x\notin\cL(B)$.
Let $G_k$ be the element of $\cG'$ that contains the leaf labeled $\cL(a)$. Clearly, $G_k$ also contains the leaf labeled $\cL(c)$ and the leaf whose label contains $b'$. Furthermore, since $\cG$ is a super-forest for $\cF$, note that $G_k$ contains a leaf whose label contains $x$. Furthermore, by the definition of $e_B$, we have that the $\mrca_{G_k}(b',x)$ lies on the path from the leaf  labeled $a'$  to the leaf  labeled $c'$  in $G_k$ and, therefore, the $\mrca_{T}(b',x)$ lies on the path from the leaf labeled $\cL(a)$ to the leaf labeled $\cL(c)$ in $T$. Now, it is easily checked that $F_j$ and $F_k$ are not vertex-disjoint in $T$; thereby contradicting that $\cF$ is an agreement forest for $S$ and $T$.

Third, assume that $\cF$ contains two components 
$F_j$ and $F_k$ such that $\cL(a)\subset \cL(F_j)$ and $\cL(c)\subset \cL(F_k)$. 
Hence, there exist elements  $x,y \in \cX$ such that  
$x\in \cL(F_j)$, $x\notin \cL(a)$, $y\in \cL(F_k)$, and $y\notin \cL(c)$. Note that $x$ or $y$ may or may not be elements of $\cL(B)$. In this case, it is straightforward to see that we can  derive the same contradiction as in Case (1).\\ 

By combining Cases (1) and (2), we deduce that there exists a super-forest of $\cF$ that can be constructed from $\cG'$ by deleting one of $\{e_a,e_c\}$ if $a\nsim_{\cG'} c$ or one of  $\{e_a,e_B,e_c\}$ if $a\sim_{\cG'} c$. Thus, this super-forest is an element of $\cG_{l+1}$. \\


Second, suppose that $\{a,c\}$ is a common cherry for $R$ and $\cG'$. Again, to derive a contradiction, assume that $\cG_{l+1}$ does not contain any reduced forest that is a super-forest of $\cF$. In particular, 
this implies that $e_a$, $e_c$, and $\{a,c\}$ are all bad choices for $\cG'$ and $\cF$. Thus, similar to Case (1),  $\cF$ contains two distinct components $F_j$ and $F_k$  such that $\cL(a)\subset\cL(F_j)$ and $\cL(c)\subset\cL(F_k)$, and for which  there exist elements $x,y \in \cX$ such that  $x\in \cL(F_j)$, $x\notin \cL(a)$, $y\in \cL(F_k)$, and $y\notin \cL(c)$. Applying the same argument as in Case (1), this contradicts that the elements of $\cF$ are vertex-disjoint in $S$. Thus, one of $e_a$, $e_c$, or $\{a,c\}$ is not a bad choice for $\cG'$ and $\cF$. If $e_a$ or $e_c$ is not a bad choice for $\cG'$ and $\cF$, then $\cG_{l+1}$ clearly contains a forest that is a super-forest for $\cF$. On the other hand, if  $e_a$ and $e_c$ are  both bad choices for $\cG'$ and $\cF$, then $\{a,c\}$ is not such a choice. Hence, calling {\sc cherryReduction}$(R,\cG',M,\{a,c\})$ returns a forest $\cG''$ such that $\cG''$ is a super-forest for $\cF$. Note that the underlying forest of $\cG''$ is $\cG$. Since $|\cG|<|\cF|$, it follows from Lemma  \ref{lemma:cuttingCorrectly}, that after some additional recursions of {\sc allMAAFs}, the algorithm chooses a cherry in line 10 of the pseudocode of {\sc allMAAFs} and subsequently deletes an edge in order to obtain a forest of size $|\cG|+1$. Then by applying the arguments of Cases (1) and (2), and the argument of this paragraph (depending on the type of cherry the algorithm has chosen), it is easily checked that $\cG_{l+1}$ contains a forest that is a super-forest of $\cF$. This completes the proof of the theorem.
\end{proof}

\section{Running time of the algorithm}\label{sec:RT}
In this section, we detail the running time of the algorithm  {\sc allMAAFs}.

%

\begin{theorem}\label{t:rt}
Let $S$ and $T$ be two rooted binary phylogenetic $\cX$-trees, and let $k$ be an integer.The running time of {\sc allMAAFs}$(S,T,S,T,k,\emptyset)$ is $O(3^{|\cX|})$.
\end{theorem}

\begin{proof}
 Let $\cF$ be a forest for $T$ that has been obtained from $T$ by deleting $n$ edges. Recall that {\sc allMAAFs} stops when $\overline{\cF}=\emptyset$ (see Remark \ref{rem:onlyOneLeafF}).
  It is easy to see that, $|\cX| - n-1$ cherry reductions are needed to reduce $\cF$ to a forest, say $\cF'$, such that $\overline{\cF'}=\emptyset$. 
 Thus, the number of recursive calls is $O(|\cX|)$.  
Since {\sc allMAAFs} is called for at most 3 times from within each recursion, it now follows that the running time of {\sc allMAAFs}$(S,T,R,\cF,k,M)$ is $O(3^{|\cX|})$ as claimed.
\end{proof}

While the worst-case running time that is presented in Theorem~\ref{t:rt} is purely theoretical, it can be significantly optimized in the following way. Bordewich and Semple~\cite{sempbordfpt2007} showed that the problem of calculating the minimum number of hybridization events that is needed to simultaneously explain two rooted binary phylogenetic $\cX$-trees $S$ and $T$ is fixed-parameter tractable. They used two reductions---called the {\it subtree and chain reduction}---to establish this result. Loosely speaking, these reductions replace different types of features that are common to $S$ and $T$ with a small number of new leaves, thereby shrinking the original trees to their respective cores while preserving their hybridization number in a well-defined way. In fact, these two reductions are sufficient to yield a kernelization of the above-mentioned problem. More precisely, it is shown in~\cite[Lemma 3.3]{sempbordfpt2007} that, by repeatedly applying the subtree and chain reductions to $S$ and $T$ until no further reduction is possible, the leaf set size of the so-obtained rooted binary phylogenetic trees is at most $14h(S,T)$. 
It is now straightforward to see that modifying {\sc allMAAFs}$(S,T,R,\cF,k,M)$  in the following way is sufficient to make use of this result. 
\begin{enumerate}
\item If $R=S$ and $\cF=T$, apply the subtree and chain reduction until no further reduction is possible and directly return $(\emptyset,k-1)$ if the leaf set size of the obtained trees is greater than $14k$.
\item Introduce a new global variable, say $w$, that is used to keep track of the weight of each initially reduced common chain of $S$ and $T$ (for details, see~\cite{sempbordfpt2007}). Additionally, whenever {\sc cherryExpansion} is called for a forest throughout a run of {\sc allMAAFs}, also call {\sc subtreeExpansion} and {\sc chainExpansion} to reverse each initially performed subtree and chain reduction, respectively. 
\item For each potential acyclic-agreement forest $\cF'$ for $S$ and $T$ that is returned from calling {\sc cherryExpansion}, {\sc subtreeExpansion}, and {\sc chainExpansion} (see line 4 of the pseudocode of {\sc allMAAFs}),  do not only check if {$\cF'$} is acyclic, but also whether or not it is a so-called legitimate-agreement forest (for details, see~\cite{sempbordfpt2007}). Note that this additional check can be performed in polynomial time. 
\end{enumerate}
\noindent We denote this extended version by {\sc allMAAFs*}$(S,T,S,T,k,\emptyset)$.\\

Now, noting that the subtree and chain reduction can be computed in $O(n^3)$ for two rooted binary phylogenetic $\cX$-trees, where $n=|\cX|$~\cite{sempbordfpt2007}, the next corollary is an immediate consequence of Theorem~\ref{t:rt}, and the kernelization ideas that are presented in~\cite{sempbordfpt2007} and briefly summarized prior to this paragraph.
\begin{corollary}
Let $S$ and $T$ be two rooted binary phylogenetic $\cX$-trees, and let $k$ be an integer.The running time of {\sc allMAAFs*}$(S,T,S,T,k,\emptyset)$ is $O(3^{14k}+n^3)$, where  $n=|X|$.
\end{corollary} 

\section{Conclusions}\label{sec:conclu}
A topical question in current mathematical research on reticulate evolution is how to construct {\it all} rooted 
hybridization networks that display a pair of rooted binary phylogenetic trees such that the number of 
hybridization vertices is minimized.  In this paper, we have made
a first step towards achieving this goal by developing the first non-naive
algorithm---called {\sc allMAAFs}---that computes all maximum-acyclic-agreement forests for two
rooted binary phylogenetic trees  on the same taxa set. {While this paper describes the theoretical framework of {\sc allMAAFs} and establishes the algorithm's correctness, a practical implementation is published in a forthcoming paper~\cite{ASH} and freely available as part of {\sc Dendroscope}~\cite{Huson2007Dendroscope}.} Note that despite the worst-case running time of {\sc allMAAFs}, the algorithm seems to perform well in practice \cite{ASH} for simulated and biological data sets. It is part of ongoing research to
extend 
the algorithm {\sc HybridPhylogeny}~\cite{BSS06} in order to compute \emph{all}
possible hybridization networks that display a pair of rooted phylogenetic trees and whose number of hybridization vertices is minimized for when a maximum-acyclic-agreement
forest for these two trees given. In combination with {\sc allMAAFs}, such an algorithm will then compute all possible
minimum hybridization networks that display a pair of phylogenetic trees.

\section*{Acknowledgements}
We thank Daniel Huson for helpful discussions. Financial support from the University of T\"ubingen is gratefully acknowledged.

\bibliographystyle{plain} 
\bibliography{bibliographyleo}

\end{document}